\documentclass[letterpaper,twocolumn,10pt]{article}
\usepackage{usenix2019_v3}

\usepackage{blindtext}
\usepackage{amsfonts}
\usepackage{tabularx}
\usepackage{amsmath}
\usepackage{amsthm}
\usepackage{thmtools}
\usepackage{thm-restate}
\usepackage[framemethod=tikz]{mdframed}
\usepackage[]{algorithm2e}
\usepackage{etoolbox}
\usepackage{booktabs}
\usepackage{siunitx}
\usepackage{float}

\usepackage{pifont} % for xmarks

\usepackage{enumitem}
\setlist{noitemsep}

\usepackage{footnote}
\makesavenoteenv{tabular}
\makesavenoteenv{table}

\usepackage{hyperref} 
\usepackage{cleveref}

\usepackage{tikz}
\tikzset{font=\footnotesize}
\usetikzlibrary{positioning}
\usetikzlibrary{shapes}
\usetikzlibrary{patterns}
\usetikzlibrary{fit}
\usetikzlibrary{matrix}
\usetikzlibrary{calc,patterns,intersections,decorations.pathreplacing}

\usepackage{array}
\usepackage{pgfplots}
\usepackage{pgfplotstable}
\usepackage[
	n,
	operators,
	advantage,
	sets,
	adversary,
	landau,
	probability,
	notions,	
	logic,
	ff,
	mm,
	primitives,
	events,
	complexity,
	asymptotics,
	keys]{cryptocode}

\usepackage[scaled=0.88]{helvet}  % comment out this instruction to suspend scaling

\usepackage{enumitem}

\usepackage{multirow}

\newtheorem{theorem}{Theorem}

\theoremstyle{definition}
\newtheorem{definition}{Definition}

\newtheorem{proc}{Procedure}

\newcounter{CommentCounter}

\definecolor{darkspringgreen}{rgb}{0.09, 0.45, 0.27}
\definecolor{4}{RGB}{129,15,124}
\definecolor{5}{RGB}{140,150,198}
\definecolor{6}{RGB}{191,211,230}

\pgfplotscreateplotcyclelist{colorbrewer-blues}{
{red},
{blue},
{darkspringgreen},
{4},
{5},
{6}
}

\pgfplotsset{
	every semilogx axis/.append style={
		cycle list name=colorbrewer-blues,
		ylabel near ticks,
		no markers,
		every axis plot/.append style={very thick},
		legend cell align=left,
		legend style={draw=none,at={(0,1)},	anchor=north west}
	}
}

\DeclareMathOperator*{\argmax}{argmax}

\newmdenv[
    innerlinewidth=0.5pt,
    roundcorner=4pt,
    linecolor=red,
    innerleftmargin=6pt,
    innerrightmargin=6pt,
    innertopmargin=6pt,
    innerbottommargin=6pt]{todobox}

\newcounter{scheme}

\newcommand{\para}[1]{\vspace{1mm}\noindent\textbf{#1}} % starts a paragraph with bold title
\newcommand{\parait}[1]{\vspace{1mm}\noindent\textit{#1}} % starts a paragraph with italics

\pgfplotstableset{
    col sep=comma,
    every head row/.style={
      before row=\toprule,after row=\midrule,
    },
    every last row/.style={
      after row=\bottomrule
    },
}

\newcommand{\name}{\textsc{VoteAgain}\xspace}

\newcommand{\group}{\mathbb{G}}
\newcommand{\generator}{g}
\newcommand{\grouporder}{p}
\newcommand{\Zp}{\mathbb{Z}_{\grouporder}}

\newcommand{\randin}{\in_R}

\renewcommand{\secpar}{\ell} % Already in cryptocode

\newcommand{\SPK}{\textsf{SPK}}

\newcommand{\mixdecprotocol}{\code{Vote.MixDecryptTally}}

\newcommand{\eclabel}{EC}
\newcommand{\eckeygen}{\code{\eclabel.KeyGen}}
\newcommand{\ecenc}{\code{\eclabel.Enc}}
\newcommand{\ecdec}{\code{\eclabel.Dec}}

\newcommand{\ctxt}{c}

\newcommand{\lastvoter}{n}
\newcommand{\lasttrustee}{t}
\newcommand{\threshold}{k}
\newcommand{\lastballot}{n_B}
\newcommand{\lastdummy}{n_D}

\newcommand{\vid}{vid}
\newcommand{\candidate}{c}
\newcommand{\ballot}{\beta}
\newcommand{\ballotnumber}{m}
\newcommand{\voteproof}{\pi}

\newcommand{\encryptedvote}{\code{v}}
\newcommand{\encryptedvid}{\gamma}
\newcommand{\encryptedindex}{I}

\newcommand{\pa}{\text{PA}\xspace}
\newcommand{\ts}{\text{TS}\xspace}
\newcommand{\pbb}{\text{PBB}\xspace}

\newcommand{\code}[1]{\textsf{#1}\xspace}
\newcommand{\setup}{\code{Setup}}

\newcommand{\gettoken}{\code{GetToken}}
\newcommand{\vote}{\code{Vote}}
\newcommand{\valid}{\code{Valid}}
\newcommand{\filter}{\code{Filter}}
\newcommand{\verifyfilter}{\code{VerifyFilter}}
\newcommand{\tally}{\code{Tally}}
\newcommand{\verifytally}{\code{VerifyTally}}
\renewcommand{\verify}{\code{Verify}}

\renewcommand{\adv}{\ensuremath{\mathcal{A}}\xspace} % Already in cryptocode
\newcommand{\nrcorrupted}{n_C}

\newcommand{\gamebit}{b}

\newcommand{\voter}{\mathcal{V}}

\newcommand{\candidatelist}{\mathcal{C}}
\newcommand{\electoralroll}{\mathcal{E}}

\newcommand{\tstag}{\theta}

\newcommand{\ballotgroup}{G}
\newcommand{\nrballotgroups}{\kappa}
\newcommand{\nrballotsingroup}{\chi}

\newcommand{\verifbexp}[3]{\code{Exp}^{\code{ver},#1}_{#2,#3}}
\newcommand{\oraclevotesimple}{\oraclesym\code{vote}}

\newcommand{\partialtally}{\bar{\rho}}
\newcommand{\partialoperator}{\star_R}

\newcommand{\corruptedset}{\code{C}}
\newcommand{\checked}{\code{Verified}}

\newcommand{\nrbad}{n_B}
\newcommand{\nrunchecked}{n_U}
\newcommand{\nrverified}{n_V}

\newcommand{\corruptionevents}{\code{C}}
\newcommand{\honestvotes}{\code{HVote}}
\newcommand{\nrtokens}[1]{\code{\#tokens}(#1)}
\newcommand{\ctr}{\code{ctr}}

\newcommand{\nrchecks}{\lambda}

\newcommand{\TS}{\textsf{TS}\xspace}
\newcommand{\PA}{\textsf{PA}\xspace}

\renewcommand{\pk}{\code{pk}}
\renewcommand{\sk}{\code{sk}}

\newcommand{\tspk}{\pk_{\code{\TS}}}
\newcommand{\tssk}{\sk_{\code{\TS}}}
\newcommand{\selectedvotelistdummies}{\selectedvotelist_D}

\newcommand{\openedrandomizers}{R}

\newcommand{\papk}{\pk_{\code{\PA}}}
\newcommand{\pask}{\sk_{\code{\PA}}}

\newcommand{\votepk}{\pk_{\code{T}}}
\newcommand{\votesk}{\sk_{\code{T}}}
\newcommand{\voteski}{\sk_{\code{T},i}}
\newcommand{\nrcandidates}{n_C}

\newcommand{\votekeygen}{\code{Vote.DKeyGen}}
\newcommand{\votekeygensimple}{\code{Vote.KeyGen}}
\newcommand{\voteenc}{\code{Vote.Enc}}
\newcommand{\voteverify}{\code{Vote.Verify}}
\newcommand{\votedec}{\code{Vote.Dec}}
\newcommand{\voteverifytally}{\code{Vote.VerifyTally}}
\newcommand{\votezeroenc}{\code{Vote.ZEnc}}

\newcommand{\result}{r}

\newcommand{\tallyproofdec}{\Pi}

\newcommand{\decryptedvid}[1]{\overline{\vid}_{#1}}
\newcommand{\decryptedindex}[1]{\overline{\ballotnumber}_{#1}}

\newcommand{\selectedvotelist}{\mathcal{S}}
\newcommand{\selectedvote}[1]{V_{#1}}
\newcommand{\preselectedvote}[1]{\overline{V}_{#1}}

\newcommand{\shuffleproof}{\pi_{\sigma}}

\newcommand{\decryptproof}{\pi^{\textsf{dec}}}

\newcommand{\reencryptionproof}{\pi^{\textsf{sel}}}

\newcommand{\resultfunction}{\rho}
\newcommand{\vidspace}{\group}

\newcommand{\resultspace}{\mathbb{R}}

\newcommand{\extractoralgorithm}{\code{Extract}}
\newcommand{\validindividual}{\code{ValidInd}}

\newcommand{\strippedballot}{\ballot'}
\newcommand{\shuffledballot}{\ballot''}

\newcommand{\strippedballotlist}{B'}
\newcommand{\shuffledballotlist}{B''}
\newcommand{\decryptedvidindexlist}{C}
\newcommand{\filteredvotes}{F}

\newcommand{\votingscheme}{\mathcal{V}}
\newcommand{\bprivexp}[3]{\code{Exp}^{\code{bpriv},#1}_{#2,#3}}

\newcommand{\simfilter}{\code{SimFilter}}
\newcommand{\simproof}{\code{SimTally}}

\newcommand{\oraclesym}{\mathcal{O}}
\newcommand{\oracles}{\oraclesym}
\newcommand{\oraclevote}{\oraclesym\code{voteLR}}

\newcommand{\oraclecast}{\oraclesym\code{cast}}
\newcommand{\oracleboard}{\oraclesym\code{board}}

\newcommand{\oracletally}{\oraclesym\code{tally}}

\newcommand{\filteradditions}{\Phi}

\newcommand{\BB}{\code{BB}}

\newcommand{\nmcpaexp}[1]{\code{Exp}^{\code{nm-cpa},#1}_{\adv}}
\newcommand{\oracledec}{\oraclesym\code{dec}}

\newcommand{\nroraclevotes}{n_{\mathcal{O}}}

\newcommand{\sconsexp}{\code{Exp}^{\code{s-cons}}_{\adv, \mathcal{V}}}

\newcommand{\votetoken}{\tau}
\newcommand{\tokensignature}{\signature^{\votetoken}}
\newcommand{\tokencounter}{m}

\newcommand{\nrvoters}{N}

\newcommand{\dummyindices}{\mathcal{D}}

\newcommand{\signkeygen}{\code{Sig.Keygen}}
\newcommand{\signsign}{\code{Sig.Sign}}
\newcommand{\signverify}{\code{Sig.Verify}}

\newcommand{\pksign}{\pk_{\sigma}}
\newcommand{\sksign}{\sk_{\sigma}}

\newcommand{\signature}{\sigma}

\newcommand{\cover}{\mathfrak{C}}

\newcommand{\coversizebase}{k}

\newcommand{\coversize}{|\cover|}
\newcommand{\groupsize}{s}
\newcommand{\nrvotersgroup}{z}

\newcommand{\nrcastvoters}{\nu}

\newcommand{\scorrexp}{\code{Exp}^{\code{s-corr}}_{\adv, \mathcal{V}}}
		
\begin{document}

\date{}
\title{\Large \bf \name: A scalable coercion-resistant voting system}
\author{
{\rm Wouter Lueks} \\
EPFL, SPRING Lab
\and
{\rm I\~{n}igo Querejeta-Azurmendi} \\
Universidad Carlos III Madrid \\
ITFI, CSIC
\and 
{\rm Carmela Troncoso} \\
EPFL, SPRING Lab
}

\maketitle

% TODO: remove for camera-ready, but force page numbers for now
% Force page numbering
\thispagestyle{plain}
\pagestyle{plain}

\begin{abstract}
  The strongest threat model for voting systems considers
  \emph{coercion resistance}: protection against coercers that force 
  voters to modify their votes, or to abstain.
  Existing remote voting systems either do not provide this property; 
  require an expensive tallying phase; or burden users with the need to store
  cryptographic key material and with the responsibility to deceive their coercers.
  We propose \name, a scalable voting scheme that relies on the 
  revoting paradigm to provide coercion resistance.
  \name uses a novel deterministic ballot padding
  mechanism to ensure that coercers cannot see whether a vote has been replaced. 
  This mechanism ensures tallies take \emph{quasilinear time}, 
  making \name the first revoting
  scheme that can handle elections with millions of voters. We prove that \name provides 
  ballot privacy, coercion resistance, and verifiability; and we demonstrate its
  scalability using a prototype implementation of its core cryptographic primitives.
\end{abstract}

% !TEX root = ../open_filtering.tex
\section{Introduction}
  \begin{table*}[htbp]
  \centering
{\small
  \caption{\label{tab:re-voting-settings} Comparison of different voting
    schemes.}
  \begin{tabular}{@{}llllcllll@{}}
    \toprule
    & \multicolumn{3}{c}{\footnotesize{Revoting}} & & & \multicolumn{3}{c}{\footnotesize{Security Properties}} \\
    \cmidrule{2-4} \cmidrule{7-9}
    & \rotatebox[origin=c]{60}{\footnotesize{Deniable}}
    & \rotatebox[origin=c]{60}{\footnotesize{Verif. Filter}}
    & \rotatebox[origin=c]{60}{\footnotesize{Complexity}}
    & \rotatebox[origin=c]{60}{\footnotesize{Crypto state}}
    & \rotatebox[origin=c]{60}{\footnotesize{Authentication}}
    & \rotatebox[origin=c]{60}{\footnotesize{Ballot Privacy}}
    & \rotatebox[origin=c]{60}{\footnotesize{Verifiability}}
    & \rotatebox[origin=c]{60}{\footnotesize{Coercion Res.}}\\
    \midrule
    % WARNING: footnote text is manually placed, make sure it ends up in the
    % right place when moving the table.
    JCJ \cite{Juels:2005:CEE:1102199.1102213, Clarkson2008, Bursuc2012} & No\footnotemark
                                            & Yes & $n^2$ & Yes & $\threshold$-out-of-$\lasttrustee$ & $\threshold$-out-of-$\lasttrustee$ & $\threshold$-out-of-$\lasttrustee$ & $\threshold$-out-of-$\lasttrustee$ + AC\\
    Black-box \cite{Gjosteen_analysisof} & TTP & No & $n$ & Yes & Unclear & $\threshold$-out-of-$\lasttrustee$ & TTP & TTP \\
    Revote \cite{Locher2016, 193462} & $\threshold$-out-of-$\lasttrustee$ & Yes & $n^2$ & Yes & TTP & $\threshold$-out-of-$\lasttrustee$ & TTP& $\threshold$-out-of-$\lasttrustee$ + AC \\
    Helios \cite{Adida:2008:HWO:1496711.1496734} & \multicolumn{3}{c}{\emph{revoting is not possible}} & No & TTP & $\threshold$-out-of-$\lasttrustee$ & TTP & N/A \\
    \textbf{\name} & TTP & Yes & $n \log n$ & No & TTP & $\threshold$-out-of-$\lasttrustee$ & TTP & TTP\\
    \bottomrule
  \end{tabular}
}
\end{table*}

Remote electronic voting, i.e., voting outside a poll-booth environment, in which
voters cast their ballot from their devices is susceptible to large-scale vote buying and
coercion~\cite{Juels:2005:CEE:1102199.1102213}.
Yet, many deployed electronic voting systems~\cite{Adida:2008:HWO:1496711.1496734,Scytl,HaenniKLD17}
do not support coercion resistance. This might be suitable in Western democracies
where freedom and privacy are well rooted in society.
However, under more authoritarian regimes~\cite{PakistanReport} or in younger democracies~\cite{NorrisWC18},
coercion is a serious problem.

There are two kind of coercion-resistant electronic voting systems in the literature.
The first kind provides users with fake voting credentials that voters use/produce when coerced, enabling deletion of coerced votes~\cite{Juels:2005:CEE:1102199.1102213,Clarkson2008}.
This approach has several downsides: (i) voters need to store their true voting credential on their devices,
(ii) the system cannot give feedback on whether the correct credential was used, and thus voters cannot be sure if their vote has been recorded correctly at the time of voting,
and (iii) voters need to convincingly lie while being coerced which may be a challenge.
The second kind relies on the \emph{revoting paradigm}~\cite{Grimm2006MCI,Post10,193462,Locher2016}.
These schemes avoid the drawbacks associated with the fake-credential
approach by allowing voters to submit fully to coercers, and instead enabling them to
supersede coerced votes by casting a new ballot. This approach requires that the coercer cannot detect whether
a voter has cast new ballots. To achieve this,
state-of-the-art schemes~\cite{193462,Locher2016} \emph{require a quadratic
number of operations}, concretely a pair-wise comparison of all ballots,
to privately filter superseded ballots.
As an example, for the Iowa Democratic caucus with only 176,574 voters,
Achenbach et al.'s solution~\cite{193462} would require 1.8 core years to filter the ballots.

We propose \name, a scalable coercion-resistant (re)voting scheme.
\name's efficiency relies on two key insights:
First, one can hide the number of ballots per user by
\emph{inserting a deterministic number of dummy ballots which
depends solely on the number of voters and the number
of cast ballots}. Thus, it reveals nothing about the number
of ballots cast by individual voters, hiding
any (re)voting patterns induced by voters or coercers.
Second, because of the deterministic nature of the approach
one can \emph{execute filtering in the clear}, reducing
the filtering time \emph{from quadratic to quasilinear}:
$O(n \log n)$ where $n$ is the number of
ballots. As a result, for the Iowa caucus
\emph{our construction requires under 14 core minutes}. We estimate that
\name using 224 cores (less than \$50 on Amazon, or \$75K on dedicated hardware)
can filter hundreds of millions of ballots in hours.

% Price computation:
%
% Cloud approach: r5.metal (96 cores, 768GiB): $6.048/hour -> $50 max
% Dell R840 Rack Server: (2x Xeon Gold 5220, 40 cores)
% Dell R840 Rack Server: (4x Xeon Gold 5220 20C, total 80 cores, 512GB) $23430

We make the following contributions:

\para{\checkmark} We introduce \name, a novel remote electronic revoting scheme based on well defined and widely used cryptographic constructions.

\para{\checkmark} We introduce a novel efficient deterministic padding scheme that hides revoting at a low cost. The complexity of the resulting filtering phase is $O(n \log n)$ where $n$ is the number of ballots. Our experiments show that in many practical scenarios the cost can be even lower.

\para{\checkmark} We show that previous definitions of coercion resistance in the revoting setting are vacuous.
We provide a new coercion-resistance definition and we adapt modern definitions of ballot privacy~\cite{Bernhard:2015:SCA:2867539.2867668} and verifiability~\cite{CortierGKMT16,CortierGGI14} to the revoting setting. We prove that \name satisfies these definitions.

\para{\checkmark} We evaluate the scalability of \name on a prototype
implementation of the core cryptographic primitives. Our results show that \name can support elections with millions of users.

%WARNING: FOOTNOTE OF TABLE ONE, MAKE SURE TO MOOVE IF TABLE MOVES
\footnotetext[1]{Revoting is possible, but revotes are not deniable. JCJ instead achieves
  coercion resistance using fake authentication credentials.}

% !TEX root = ../open_filtering.tex

\section{Related Work}

Coercion-resistant voting schemes fall under two categories: either they 
enable voters to generate \emph{fake authentication credentials} or they allow the voter 
to \emph{revote}. Coercion resistant
schemes using fake credentials, introduced by Juels et al.~\cite{Juels:2005:CEE:1102199.1102213} (JCJ), 
are used in several voting schemes~\cite{Clarkson2008,Selections,Bursuc2012,DBLP:conf/fc/AraujoBBT16}. 
In these schemes, the voter has both real and fake authentication credentials (or pre-registered passwords
and panic passwords~\cite{Selections,EssexCH12}).
When coerced, the voter lies to the coercer, using a fake authentication credential (or handling it to the coercer), 
resulting in a non-counted ballot. Ballots cast with the real credential are counted.
These schemes provide the real authentication credential
to the voter during registration phase 
(in which the coercer \textit{must} be absent). The voter must securely store this authetication credentials
for later use, i.e., voters need to maintain cryptographic state.

Coercion resistant schemes based on revoting allow voters to cast multiple ballots and 
then filter these ballots, typically counting the last ballot per voter. 
For such a scheme to be coercion resistant, the filtering stage must be \textit{deniable}~\cite{193462},
i.e., it must not expose which ballots are filtered, as this would expose revoting actions. 
\textit{Black box} filtering where a trusted third party (TTP) performs the filtering privately 
is deniable~\cite{Gjosteen_analysisof}, but not verifiable.
To the best of our knowledge, there exist two publicly-verifiable deniable re-voting schemes~\cite{193462,Locher2016}. 
To obtain public verifiability these schemes use a distributed authority
to compare each pair of ballots (i.e., $O(n^2)$ operations) before 
shuffling to privately mark superseded ballots.
After shuffling, these marks are decrypted and the tallying server verifiably
filters superseded ballots.
As literally specified in these papers, these schemes are \emph{`not efficient
  for large scale elections'}. We confirm in Section
\ref{sec:evaluation} that Achenbach et al.'s scheme~\cite{193462}
cannot efficiently handle small elections of a hundred thousand users. 

Both the JCJ based and the private revoting based schemes offer a 
solution with a $\threshold$-out-of-$\lasttrustee$ assumption for 
coercion resistance. However, on top of that, these schemes require the
existence of Annonymous Channels (AC) to avoid coercion attacks such as 
forced abstention. 

For authentication, most schemes require users to store cryptographic
state~\cite{Juels:2005:CEE:1102199.1102213,Clarkson2008,193462,Locher2016,Bursuc2012,Gjosteen_analysisof,Ryan2015},
or remember special passwords~\cite{Selections,EssexCH12}.
Helios~\cite{Adida:2008:HWO:1496711.1496734} and
Apollo~\cite{DBLP:journals/iacr/GawelKVWZ16} rely on regular username/password.
To improve verifiability (by distributing the trust of the entity deciding which users are eligible voters), some schemes require that
voters authenticate to $\threshold$ out of $\lasttrustee$ parties~\cite{Juels:2005:CEE:1102199.1102213,Clarkson2008,Bursuc2012,DBLP:conf/fc/AraujoBBT16}.
However, this results in a complex registration phase for the user where, additionally, the
coercer is assumed to be absent. Revoting based schemes (including \name) can be extended
to this setting to reduce the trust assumptions required for authentication correctness (and 
hence verifiability).
Table \ref{tab:re-voting-settings} 
summarizes the comparison between \name and previous work. 
% Table is in intro section
% !TEX root = ../open_filtering.tex
\section{System and threat model}	

\para{Actors.}
There are five actors in \name: voters, a polling authority,
a bulletin board, a tally server, and trustees.

\parait{Voters} Voters interact with the polling authority and the public bulletin
  board to cast their ballots. Each voter has the means to
  authenticate herself to the polling authority (e.g., an electronic
  identity card). There are $\lastvoter$ voters.

\parait{Polling Authority (\pa)} The \pa authenticates users and
  provides them with ephemeral voting tokens. Voters use these tokens to sign
  their ballots before posting them to the public bulletin board.

\parait{Public Bulletin Board (\pbb)} The \pbb is an append-only list of 
cast ballots. Ballots are posted during the election phase by the voters.
  During the tally phase, the tally server and trustees post their proofs 
  and results to the bulletin board.
  Ad-hoc implementations~\cite{Heather2009} or 
  blockchain-based implementations~\cite{DBLP:conf/ndss/Garman0M14,Fromknecht2014ADP}
  would be suitable for our \pbb.

\parait{Tally Server (\ts)} The \ts filters the ballots.
  It adds dummy ballots, shuffles the ballots, groups them by
  voter, and selects the last ballot for each voter.

\parait{Trustees} The trustees mix and decrypt the selected ballots to reveal the outcome of the election. Each
trustee has a partial decryption key for a $\threshold$-out-of-$\lasttrustee$
encryption system. 

\newcommand{\xmark}{\ding{55}}%
\begin{table}[t]
  \centering
{\small
  \caption{\label{tab:coercion-differences} Comparison of assumptions in
    pre-election phase and election phase required to mitigate
    coercion attacks fake credentials and revoting based systems.}
  \begin{tabular}{lcc@{}}
    \toprule
    \multicolumn{1}{@{}l}{Assumptions} & Fake Credentials & Revoting \\
    \midrule
    \multicolumn{1}{@{}l}{\emph{Pre-election phase}} & & \\[2ex]
    No coercion & \checkmark & N/A \\[2ex]
    Inalienable authentication & \checkmark & N/A \\[2ex]
    \multicolumn{1}{@{}l}{\emph{Election phase}} & & \\[2ex]
    Lie convincingly & \checkmark & \xmark \\[2ex]
    Coercer absent some point & \multirow{ 2}{*}{\checkmark} & \multirow{ 2}{*}{\checkmark} \\
    during election & &\\[2ex]
    Absence of coercer after & \multirow{ 2}{*}{\xmark} & \multirow{ 2}{*}{\checkmark} \\
    coercion & & \\[2ex]
    Device holding voting secrets or & \multirow{ 2}{*}{\checkmark} & \multirow{ 2}{*}{\xmark}  \\
    need to remember special pwds & & \\[2ex]
    Inalienable authentication & \xmark & \checkmark \\ 
    % WARNING: footnote text is manually placed, make sure it ends up in the
    % right place when moving the table.
    \bottomrule
  \end{tabular}
}
\end{table}

\para{Threat model.}\label{definition}
We assume an adversary \adv whose goal it is to coerce voters
into casting votes for a particular candidate or to abstain. 
This adversary, although computationally bounded, 
may coerce any voter -- but not \emph{all} voters. 
Under coercion, the coerced voter does exactly as instructed (without
needing to lie).
The coercer learns all information stored and received by the voter at the time of coercion.
We assume that after coercion the coercer does not control a voter for some time before the end
of the election, such that the voter can cast at least one more vote. 
We also assume that the user's means of authentication is \emph{inalienable}~\cite{193462}, 
that is, a coercer can neither eliminate nor duplicate a voter's means of authentication.

While these assumptions are strong, we point out that so are the
assumptions behind coercion resistant solutions
that rely on fake 
credentials~\cite{Juels:2005:CEE:1102199.1102213, Clarkson2008, Bursuc2012} (see
Table~\ref{tab:coercion-differences}). Fake-credential based solutions
assume that users cannot be coerced during registration and hence need inalienable means of 
authentication during this phase;
that users can store and hide cryptographic key material and hence are required to have 
access to where this material is stored during the voting phase;
and that users can lie convincingly. 
These assumptions are not needed in \name. Our construction 
allows users to vote from any device, preventing coercion attacks that rely
on destroying or stealing the voting device. 

In \name, voters authenticate against the PA every time before voting to obtain an ephemeral
voting token. Thus, the PA \emph{must} be honest with respect to verifiability 
and coercion resistance.
To enable quasilinear filtering we also require that the TS is honest with respect to coercion resistance. 
This assumption is stronger than Achenbach et
al.'s $\threshold$-out-of-$\lasttrustee$ assumption on the
trustees~\cite{193462}, but their relaxation comes at a quadratic computational
cost, see Table~\ref{tab:re-voting-settings}.

Finally, we require \name to satisfy the following informal properties. We
formalize them in Section~\ref{formal-sec}. 
Table~\ref{tab:adv-assumptions}
summarizes the trust required in each party for achieving each of the properties.  

\begin{definition}[Ballot privacy\cite{Bernhard:2015:SCA:2867539.2867668}]
  \label{BallSecre}
  Assuming that at least $\threshold$ trustees are honest, no coalition of
  malicious parties (including the PA and TS) can learn the vote of an honest user.
\end{definition}

\begin{definition}[Coercion resistance]
  \label{CoerResis}
  Assuming that the PA and TS are honest, no coercer can use the \pbb to determine if coercion was successful or not, provided that the election outcome does not leak this information.
\end{definition}

\begin{definition}[Verifiability]
  \label{UnivVerif}
  Assuming that the PA is honest, \name guarantees that: 
  (i) the last ballot per voter will be tallied,
  (ii) adversary \adv cannot include more malicious votes in the tally than the number of voters it controls,
  and (iii) honest ballots cannot be replaced.
  If voters do not verify that their ballots are correctly appended to the \pbb, 
 ballots can be dropped or replaced by earlier ballots if those exist.
\end{definition}

\newcommand{\udensdash}[1]{%
    \tikz[baseline=(todotted.base)]{
        \node[inner sep=1pt,outer sep=0pt] (todotted) {#1};
        \draw[densely dashed] (todotted.south west) -- (todotted.south east);
    }%
}%

\begin{table}[tbp]
  \caption{\label{tab:adv-assumptions} Trust assumptions on \name entities
  to achieve each property.}
  \centering
{\small
  \begin{tabular}{@{}lccc@{}}
    \toprule
     & Ballot Privacy & Verifiability & Coercion resistance\\
    \midrule
    \pa & Untrusted & Trusted & Trusted\\
    \ts & Untrusted & Untrusted & Trusted \\
    \pbb & Untrusted & Untrusted & Untrusted \\
    Trustees & $\threshold$-out-of-$\lasttrustee$ & Untrusted & Untrusted \\
    \bottomrule
  \end{tabular}
}
\end{table}
% !TEX root = ../open_filtering.tex

\section{\name: High-level overview}
\label{sec:key-ideas}

We sketch the key ideas of \name. For simplicity, we omit, in this 
section, the zero-knowledge proofs that parties use to show
that they performed operations correctly. We describe the
protocols in detail in Section~\ref{sec:formal-desc}.

\name proceeds in three phases: the pre-election phase, the election phase, and
the tally phase. During the pre-election phase, the polling authority (PA)
assigns to each voter $i$ a random voter identifier $\vid_i$, and a random
initial ballot index $\ballotnumber_i$. These values are known only to the PA.

\para{Casting ballots.}
During the election phase, voters can cast as many votes as they want. To cast a
vote, voter $i$ first authenticates to the PA using her inalienable
authentication means to obtain an ephemeral voting token. This voting token
contains an encrypted voter identifier $\encryptedvid$, containing $\vid_i$, and
an encrypted ballot index $\encryptedindex$, containing $\ballotnumber_i$. After
each authentication, the PA increases $\ballotnumber_i$ by one. Next, the voter
encrypts her choice of candidate as $\encryptedvote$.
Finally, the voter sends the encrypted vote $\encryptedvote$, the encrypted
voter identifier $\encryptedvid$, the encrypted ballot number $\encryptedindex$,
and a signature using the ephemeral token to the bulletin board.

\para{Filtering ballots.}
The encrypted voter identifiers and ballot indices enable the tally server (TS)
to efficiently select the last ballot for each voter. The TS uses the simplest
mechanism possible: It shuffles the ballots, and then decrypts the voter
identifiers and ballot indices. The ballots can then publicly be grouped per
voter, and the last ballot can be identified by inspection. Finally, the
trustees tally the last ballot of each voter. See
Figure~\ref{fig:overview-no-dummies}.

\para{Hiding patterns using dummies.}
By itself, shuffling and filtering is not a coercion-resistant mechanism: a coercer can
perform the 1009 attack~\cite{Smith05}. In this attack, the coercer forces a
voter to cast a specific number of ballots and looks for a group of that size in
the filtering step. If such group does not exist, the coerced voter has revoted.
In \name, the TS inserts a deterministic number of \emph{dummy ballots} and \emph{dummy
voters} before shuffling the ballots to hide such patterns while maintaining the simple public filtering procedure.

\tikzset{
  ptxt/.style={
    inner sep=0mm,
    outer sep=2mm,
    minimum width=6mm,
    minimum height=3mm,
    anchor=center,
  },
  ctxt/.style={
    draw, rectangle, ptxt,
  },
  encvid/.style={ctxt,fill=red!50},
  encidx/.style={ctxt,fill=blue!50},
  encvote/.style={ctxt,fill=green!50},
  smaller/.style={
    minimum width=3mm,
    minimum height=2mm,
  }
}

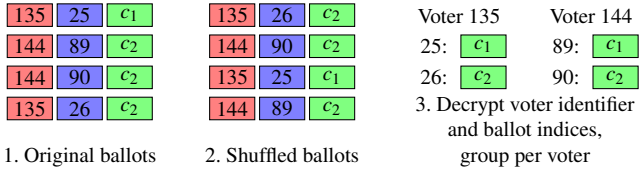
\begin{figure}
  \centering
  \begin{tikzpicture}[
      ballotgroup/.style={
        matrix of nodes,
        column sep=0.5mm,
        row sep=1mm,
      },
      belowlabel/.style={
        align=center,
        text height=5mm,
        text depth=3mm,
      }
    ]

    \matrix (group1) [ballotgroup] {
      |[encvid]| 135 & |[encidx]| 25 & |[encvote]| $c_1$ \\
      |[encvid]| 144 & |[encidx]| 89 & |[encvote]| $c_2$ \\
      |[encvid]| 144 & |[encidx]| 90 & |[encvote]| $c_2$ \\
      |[encvid]| 135 & |[encidx]| 26 & |[encvote]| $c_2$ \\
    };
    \node [belowlabel,below=-2mm of group1] {1. Original ballots};

    \matrix (group2) [ballotgroup,right=5mm of group1] {
      |[encvid]| 135 & |[encidx]| 26 & |[encvote]| $c_2$ \\
      |[encvid]| 144 & |[encidx]| 90 & |[encvote]| $c_2$ \\
      |[encvid]| 135 & |[encidx]| 25 & |[encvote]| $c_1$ \\
      |[encvid]| 144 & |[encidx]| 89 & |[encvote]| $c_2$ \\
    };
    \node (l2) [belowlabel,below=-2mm of group2] {2. Shuffled ballots};

    \matrix(group4) [ballotgroup,right=5mm of group2] {
      |[ptxt]| 25: & |[encvote]| $c_1$ \\
      |[ptxt]| 26: & |[encvote]| $c_2$ \\
    };
    \node[belowlabel,above=-4mm of group4] {Voter 135};

    \matrix(group5) [ballotgroup, right=of group4, anchor=center] {
      |[ptxt]| 89: & |[encvote]| $c_1$ \\
      |[ptxt]| 90: & |[encvote]| $c_2$ \\
    };
    \node[belowlabel,above=-4mm of group5] {Voter 144};
    \node [belowlabel,right=5mm of l2] {3. Decrypt voter identifier\\and ballot indices,\\ group per voter};
  \end{tikzpicture}
  \vspace{-10mm}
  \caption[Overview of no-dummies]{\label{fig:overview-no-dummies} Basic filtering process by tally server without using dummies.
Ballots consist of an encrypted voter identifier (\tikz{\node[encvid,
smaller] {};}), an encrypted ballot index (\tikz{\node[encidx,smaller] {};}), and
an encrypted vote (\tikz{\node[encvote,smaller] {};}).}
\end{figure}

We illustrate our dummy mechanism in Figure~\ref{fig:simple-padding}, in
a scenario with two voters (\emph{A} and \emph{B}) where, the coercer 
forces voter \emph{A} to cast 2 ballots. 
At the end of the election phase the coercer observes 4
ballots and must determine whether \emph{A} revoted (situation 2) or not (situation 1).
Without dummies, distinguishing these situations is trivial:
if \emph{A} revoted there is a group of 3 ballots and one of 1 ballot, and there are two groups of 2 ballots otherwise.
We add dummy ballots and voters to make both situations look identical.
The idea is to find a \emph{cover} of ballots that could result from both situations.
For instance, adding to either situation two dummy 
voters that cast four dummy ballots total
yields groups of 1, 2, 2, and 3 ballots. This observation makes both situations 
indistinguishable for the coercer (Figure~\ref{fig:simple-padding}, right).

To ensure that the cover is \emph{independent} from the voters' real 
actions, its appearance must depend \emph{only} on the information
available to the coercer:
(1) the number of ballots $\lastballot$ posted by users to the bulletin board; and
(2) the number of voters $\nrcastvoters$ that cast a ballot.
The goal of the dummy generation strategy is to 
allocate dummy ballots such that
the adversary observes the \emph{same cover} regardless of the
actual distributions of the $\lastballot$ ballots over $\nrcastvoters$ voters.

Consider the case of two voters, i.e., $\nrcastvoters = 2$, and 9 ballots, i.e., $\lastballot = 9$. 
As the filtering stage only reveals the sizes of the groupings and not their relation to voters
the possible adversary's observations are $(1,8), (2,7), (3,6)$, and $(4,5)$. 
To cover all these scenarios one needs 8 voters (6 of which are dummy) 
casting $1,2,3,4,5,6,7,$ and $8$ ballots, for a total of $36 - 9 = 27$ dummy ballots.

We add dummy ballots to real voters as well to reduce the number of group sizes
that are possible.
For example, in the previous scenario one can pad the cases $(1,8), (2,7), (3, 6), (4,5)$ 
to $(1,8), (2, 8), (4,8), (4,8)$. This can be covered with a cover containing
voters with $1,2,4,8$ ballots each. Building this cover requires only 2 dummy voters 
and $15 - 9 = 6$ dummy ballots. We stress that 
\emph{the number of added dummy ballots is independent of 
how the real ballots are actually distributed among the two voters}.

We refer to Section~\ref{sec:dummies} for a generic and efficient algorithm for
computing a cover.

\tikzset{
  ballot-base/.style={
    draw, rectangle,
    inner sep=0,
    outer sep=4mm,
    minimum width=1.3mm,
    minimum height=1.3mm,
  },
  ballot-box/.style={ballot-base,fill},
  ballot-dum/.style={ballot-base},
  ballot-adv/.style={ballot-base,color=gray,fill},
}

\begin{figure}
  \centering
  \begin{tikzpicture}[
      ballotgroup/.style={
        matrix of nodes,
        column sep=2mm,
        row sep=0.6mm,
      },
      toplabel/.style={
        align=center,
        text height=5mm,
        text depth=3mm,
      }
    ]

    \node (sit1) [draw, rectangle, align=left] {
      \emph{Situation 1}\\
      A: 2 ballots\\
      B: 2 ballots
    };

    \node (sit2) [draw, rectangle, align=left, below=5mm of sit1] {
      \emph{Situation 2}\\
      A: 3 ballots\\
      B: 1 ballots
    };

    \matrix (group1) [ballotgroup,right=7mm of sit1] {
      |[ballot-box]| & |[ballot-box]| \\
      |[ballot-box]| & |[ballot-box]| \\
    };

    \matrix (dummies1) [ballotgroup,right=18mm of group1.north east,anchor=north] {
      |[ballot-dum]| & |[ballot-box]| & |[ballot-box]| & |[ballot-dum]| \\
                     & |[ballot-box]| & |[ballot-box]| & |[ballot-dum]| \\
                     &                &                & |[ballot-dum]| \\
    };

    \matrix (group2) [ballotgroup,right=7mm of sit2] {
      |[ballot-box]| & |[ballot-box]| \\
                     & |[ballot-box]| \\
                     & |[ballot-box]| \\
    };

    \matrix (dummies2) [ballotgroup,right=18mm of group2.north east,anchor=north] {
      |[ballot-box]| & |[ballot-dum]| & |[ballot-dum]| & |[ballot-box]| \\
                     & |[ballot-dum]| & |[ballot-dum]| & |[ballot-box]| \\
                     &                &                & |[ballot-box]| \\
    };

    \matrix (adv1) [ballotgroup,right=18mm of dummies1.north east,anchor=north] {
      |[ballot-adv]| & |[ballot-adv]| & |[ballot-adv]| & |[ballot-adv]| \\
                     & |[ballot-adv]| & |[ballot-adv]| & |[ballot-adv]| \\
                     &                &                & |[ballot-adv]| \\
    };

    \matrix (adv2) [ballotgroup,right=18mm of dummies2.north east,anchor=north] {
      |[ballot-adv]| & |[ballot-adv]| & |[ballot-adv]| & |[ballot-adv]| \\
                     & |[ballot-adv]| & |[ballot-adv]| & |[ballot-adv]| \\
                     &                &                & |[ballot-adv]| \\
    };

    \node [toplabel,above=1mm of group1] {Original\\ballots};
    \node [toplabel,above=1mm of dummies1] {Dummy\\addition};
    \node [toplabel,above=1mm of adv1] {Cover\\(coercer observation)};
  \end{tikzpicture}
  \vspace{-5mm}
  \caption[Example of padding]{\label{fig:simple-padding} The original ballots' groups (\tikz{\node [ballot-box] {};}) create distinguishable situations. Adding 2 dummy voters casting a total of 4 dummy ballots (\tikz{\node [ballot-dum] {};}), the situations become indistinguishable.
  } 
\end{figure}
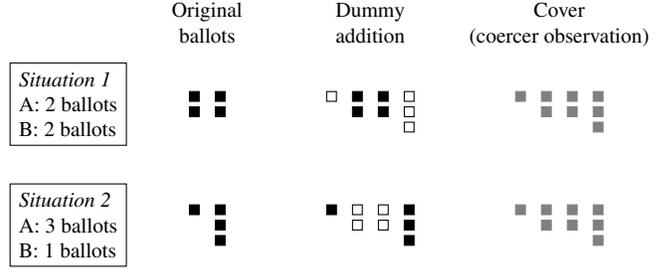

\para{Filtering with dummies.}
Before shuffling the ballots, the TS adds dummy ballots to achieve the desired
grouping. We must ensure, however, that the TS cannot modify the election
outcome. To this end, the TS tags real and dummy ballots with a different
encrypted tag.
\tikzset{
  realtag/.style={ctxt,fill=black!30},
  dummytag/.style={ctxt,fill=black!30},
}

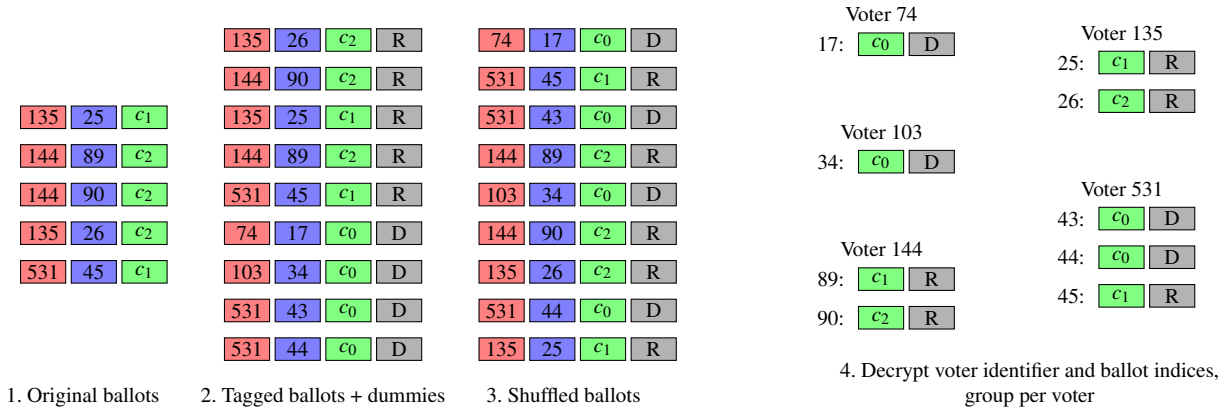
\begin{figure*}[tbp]
  \centering
  \begin{tikzpicture}[
      ballotgroup/.style={
        matrix of nodes,
        column sep=0.6mm,
        row sep=2mm,
      },
      belowlabel/.style={
        align=center,
        text height=5mm,
        text depth=3mm,
      }
    ]

    \matrix (group1) [ballotgroup] {
      |[encvid]| 135 & |[encidx]| 25 & |[encvote]| $c_1$ \\
      |[encvid]| 144 & |[encidx]| 89 & |[encvote]| $c_2$ \\
      |[encvid]| 144 & |[encidx]| 90 & |[encvote]| $c_2$ \\
      |[encvid]| 135 & |[encidx]| 26 & |[encvote]| $c_2$ \\
      |[encvid]| 531 & |[encidx]| 45 & |[encvote]| $c_1$ \\
    };

    \matrix (group2) [ballotgroup,right=5mm of group1] {
      |[encvid]| 135 & |[encidx]| 26 & |[encvote]| $c_2$ & |[realtag]| R \\
      |[encvid]| 144 & |[encidx]| 90 & |[encvote]| $c_2$ & |[realtag]| R \\
      |[encvid]| 135 & |[encidx]| 25 & |[encvote]| $c_1$ & |[realtag]| R \\
      |[encvid]| 144 & |[encidx]| 89 & |[encvote]| $c_2$ & |[realtag]| R \\
      |[encvid]| 531 & |[encidx]| 45 & |[encvote]| $c_1$ & |[dummytag]| R \\
      |[encvid]|  74 & |[encidx]| 17 & |[encvote]| $c_0$ & |[dummytag]| D \\
      |[encvid]| 103 & |[encidx]| 34 & |[encvote]| $c_0$ & |[dummytag]| D \\
      |[encvid]| 531 & |[encidx]| 43 & |[encvote]| $c_0$ & |[dummytag]| D \\
      |[encvid]| 531 & |[encidx]| 44 & |[encvote]| $c_0$ & |[dummytag]| D \\
    };
    \node (l2) [belowlabel,below=-2mm of group2] {2. Tagged ballots + dummies};
    \node [belowlabel,left=3mm of l2] {1. Original ballots};

    \matrix (group3) [ballotgroup,right=5mm of group2] {
      |[encvid]|  74 & |[encidx]| 17 & |[encvote]| $c_0$ & |[dummytag]| D \\
      |[encvid]| 531 & |[encidx]| 45 & |[encvote]| $c_1$ & |[dummytag]| R \\
      |[encvid]| 531 & |[encidx]| 43 & |[encvote]| $c_0$ & |[dummytag]| D \\
      |[encvid]| 144 & |[encidx]| 89 & |[encvote]| $c_2$ & |[realtag]| R \\
      |[encvid]| 103 & |[encidx]| 34 & |[encvote]| $c_0$ & |[dummytag]| D \\ %%
      |[encvid]| 144 & |[encidx]| 90 & |[encvote]| $c_2$ & |[realtag]| R \\
      |[encvid]| 135 & |[encidx]| 26 & |[encvote]| $c_2$ & |[realtag]| R \\
      |[encvid]| 531 & |[encidx]| 44 & |[encvote]| $c_0$ & |[dummytag]| D \\
      |[encvid]| 135 & |[encidx]| 25 & |[encvote]| $c_1$ & |[realtag]| R \\
    };
    \node [belowlabel,right=3mm of l2] {3. Shuffled ballots};

    \matrix(group4) [ballotgroup,right=15mm of group3,yshift=20mm] {
      |[ptxt]| 17: & |[encvote]| $c_0$ & |[dummytag]| D \\
    };
    \node[belowlabel,above=-4mm of group4] {Voter 74};
    
    \matrix(group8) [ballotgroup, below=of group4] {
      |[ptxt]| 34: &|[encvote]| $c_0$ & |[realtag]| D \\
    };
    \node[belowlabel, above=-4mm of group8] {Voter 103};

    \matrix(group5) [ballotgroup,right=of group4,yshift=-5mm] {
      |[ptxt]| 25: & |[encvote]| $c_1$ & |[realtag]| R \\
      |[ptxt]| 26: & |[encvote]| $c_2$ & |[realtag]| R \\
    };
    \node[belowlabel,above=-4mm of group5] {Voter 135};

    \matrix(group6) [ballotgroup, below=of group8] {
      |[ptxt]| 89: & |[encvote]| $c_1$ & |[realtag]| R \\
      |[ptxt]| 90: & |[encvote]| $c_2$ & |[realtag]| R \\
    };
    \node[belowlabel,above=-4mm of group6] {Voter 144};

    \matrix(group7) [ballotgroup, below=of group5] {
      |[ptxt]| 43: & |[encvote]| $c_0$ & |[dummytag]| D \\
      |[ptxt]| 44: & |[encvote]| $c_0$ & |[dummytag]| D \\
      |[ptxt]| 45: & |[encvote]| $c_1$ & |[dummytag]| R \\
    };
    \node[belowlabel,above=-4mm of group7] {Voter 531};

    \node [belowlabel,right=50mm of l2] {4. Decrypt voter identifier and ballot indices,\\ group per voter};
  \end{tikzpicture}
  \vspace{-5mm}
  \caption[Overview of with dummies]{\label{fig:overview-dummies} Filtering process by tally server including dummies. Labels as in
Figure~\ref{fig:overview-no-dummies}. To enable correctness proofs, the TS tags
real ballots and dummy ballots with an encrypted marker (\tikz{\node[realtag,smaller] {};}).}
\end{figure*}

To determine how to add dummies, the TS inspects the decrypted voter identifiers
and ballot indices; determines a cover; and then computes how many dummies to
add to existing voters, and how many dummies to add to dummy voters. Consider
the example in Figure~\ref{fig:overview-dummies}. Given 3 voters and 5 ballots,
a cover with groups of size 1,1,2,2, and 3 suffices. The TS therefore adds 4 dummy
ballots in step 2: 2 dummies to existing voter 531, and two dummy voters, 74 and 
103, each with one dummy vote.

After adding the dummy ballots, the TS shuffles all ballots. Next, the TS decrypts the voter identifiers and ballot indices; 
groups ballots per voter, and selects the last ballot per voter. The tags enable
the TS to prove that it did not omit real ballots cast by real voters, and it did not
count dummy votes cast by dummy voters. In particular,
the TS proves in zero-knowledge that the selected votes are either tagged as a
real vote and therefore must correspond to the last ballot of a real voter; or
the selected vote corresponds to a
dummy voter (i.e., all the ballots in the group are tagged as dummies). Finally,
the TS privately discards the selected votes corresponding to dummy voters.
We refer to Section~\ref{sec:formal-desc} for the full details.

\para{Design choices.}
%\comment[WL]{@Carmela: probably the most tendentious part of these edits}
Obtaining coercion resistance require strong assumptions on some of the parties. 
In this section we discuss our design choices and motivate
our trust assumptions (see Table~\ref{tab:coercion-differences} for a comparison with other protocols). 
First, we believe that revoting is an easy to understand solution to achieve coercion resistance. 
It requires no extra devices, no memorization, no interaction with several entities during registration, and
no lying. For instance, Estonian elections have used a revoting
model for years\cite{verifEstonia} with 44\% of the electorate having used
internet voting\cite{ivotingEstonia}. Second, it does not require voters to
securely store cryptographic material, allowing a vote cast from any device.
This further reduces the possibility of coercion attacks by confiscating the
credential storage device.

Coercion resistance \emph{requires} absence of the coercer at some point
during the process. Fake-credential solutions assume that the coercer is absent during registration
and at some point during the voting phase. Revoting, instead, assumes
that a voter will have time after the coercion to cast the last vote. 
In the case of a remote registration process, a targeted attack will
most likely succeed in both scenarios. However, attacks
scale much better in the fake-credential setting: coercers have the entire
registration period (e.g., 24 days in Spain) to coerce a voter. 
In contrast, coercers in the revoting setting 
must monitor \emph{all} coerced voters after coercion to prevent them from
revoting before the election closes.

We decide to trade-off trust with respect to coercion resistance on the \pa and \ts
to obtain high gains in usability and efficiency:
trust on the \pa relieves users
from keeping cryptographic state;
and trust on the \ts enables \name's quasilinear filtering of ballots.
% !TEX root = ../open_filtering.tex

\section{The \name voting scheme}\label{sec:voting-scheme}
\para{Preliminaries.}
Let $\secpar$ be a security parameter.
Let $\group$ be a cyclic group of prime order $\grouporder$ generated by generator $\generator$. We write $\Zp$ for the integers modulo $\grouporder$.
We write $a \randin A$ to denote that $a$ is chosen uniformly at random from the set $A$.

\name uses the ElGamal's encryption scheme given by:
A key generation algorithm $\eckeygen(\group, \generator, \grouporder)$ which outputs a public-private key-pair $(\pk = \generator^{\sk}, \sk)$ for $\sk \in_R\mathbb{Z}_p$;
an encryption function $\ecenc(\pk,m)$ which takes as input a public key $\pk$ and a message $m \in \group$ and returns a ciphertext $\ctxt = (\ctxt_1, \ctxt_2) = (\generator^{r}, m \cdot \pk^{r})$ for $r \randin \Zp$;
and an decryption algorithm $\ecdec(\sk, \ctxt)$  which returns the message $m = \ctxt_2 \cdot \ctxt_1^{-\sk}$. 
\name uses deterministic encryption (with randomness zero) as a cheap verifiable
`encoding' for the ballot tags. Because the encryption is
deterministic, verifiers can cheaply check that the encrypted tags 
have been correctly formed.

We use a traditional signature scheme given by: 
A key generation algorithm $\signkeygen(1^{\secpar})$ that generates a public-private key-pair $(\pksign,\allowbreak \sksign)$;
a signing algorithm $\signature = \signsign(\sksign, m)$ that signs messages $m \in \{0, 1\}^*$; 
and a verification algorithm $\signverify(\pksign, \signature, m)$ that outputs $\top$ if $\signature$ is a valid signature on $m$ and $\bot$ otherwise.

We use verifiable shuffles \cite{Bayer2012} to support coercion resistance in a
private way. These enable an entity to verifiably shuffle of a list of
homomorphic ciphertexts in such a way that it is infeasible for a
computationally bounded adversary to match input and output ciphertexts.

\name uses \textit{mixnets}, a standard approach~\cite{10.5555/2028012.2028018, 10.5555/647253.720294, 5272310} to compute the election result given the filtered 
ballots output by the TS.
The trustees jointly run the $\votekeygen(1^{\secpar}, \threshold, \lasttrustee, \nrcandidates)$ protocol 
where $\secpar$ is the security parameter $\secpar$, $\nrcandidates$ the number of candidates,  $\lasttrustee$ the number of trustees, and $\threshold$ is the number of trustees needed to decrypt ciphertexts. 
This protocol outputs a public encryption key $\votepk$ and each trustee $i$ obtains a private decryption key $\voteski$. 
To encrypt her vote for candidate $\candidate$, a voter calls $(\encryptedvote, \voteproof) = \voteenc(\votepk, \candidate)$ to obtain an encrypted vote $\encryptedvote$ and proof $\voteproof$ that $\encryptedvote$ encrypts a choice for a valid candidate. 
We denote the encryption of the zero candidate (i.e. no candidate) with 
explicit randomizer
 $r\randin \Zp$ by $\votezeroenc(\votepk; r)$. 
The algorithm $\voteverify(\votepk, \encryptedvote, \voteproof)$ outputs $\top$
if the encrypted vote $\encryptedvote$ is correct, and $\bot$ otherwise.
Given a list of selected votes $\{\selectedvote{1}, \ldots,
\selectedvote{\kappa}\}$, the trustees jointly run the $(\result,
\tallyproofdec) \gets \mixdecprotocol(\votepk,\allowbreak
\{\selectedvote{1},\allowbreak \allowbreak \ldots, \selectedvote{\kappa}\})$ protocol to compute the
election result $\result$ and a proof of correctness $\tallyproofdec$.
Internally, $\mixdecprotocol$ uses a standard verifiable mix network and
verifiable decryption to shuffle
and decrypt the ballots, and then computes the final result in the clear.
Any verifier can run $\voteverifytally(\votepk, \{\selectedvote{1}, \ldots,
\selectedvote{\kappa} \}, \result, \tallyproofdec)$ to verify whether the result
$\result$ is computed correctly.

\begin{table}[tbp]
{\small
  \caption{\label{tab:notation} Summary of notation.}
  \centering
  \begin{tabular}{@{}ll@{}}
    \toprule
    Symbol & Description \\
    \midrule
    $(\group, \generator, \grouporder)$ & Group, generator and prime order \\
    $\Zp$ & Integers modulo the group order $\grouporder$\\
    \midrule
    $\nrvoters$ & Number of eligible voters \\
    $\lasttrustee, \threshold$ & Number of trustees and decryption threshold \\
    $\papk, \tspk, \votepk$ & Public keys of PA, TS, and trustees \\
    $\pask, \tssk, \voteski$ & Private keys of PA, TS, and trustee $i$ \\
    \midrule
    $\vid_i, \ballotnumber_i$ & Voter identifier and ballot index of voter $i$\\
    $\encryptedvid, \encryptedindex$ & The encrypted $\vid$ and ballot index \\
    $\pk, \sk$ & Ephemeral signing keys \\
    $\votetoken, \tokensignature$ & Ephemeral voting token and signature by PA \\
    $\ballot, \signature$ & Ballot and signature using ephemeral key $\pk$\\
    $\encryptedvote, \voteproof$ & Encrypted vote and zero knowledge proof of \\
    & correct encryption \\
   \midrule 
    $\lastballot, \lastdummy$ & Number of real and dummy ballots on the board \\
    $\tstag, \tstag_R, \tstag_D$ & Ballot tags for unknown, real, and dummy ballots \\
    $\decryptedvid{i}, \decryptedindex{i}$ & Decrypted voter identifier and ballot index \\
    $\decryptproof_i, \reencryptionproof_i$ & Zero knowledge proof of correct decryption\\
    & and vote selection \\
    $\preselectedvote{i}$ & Selected vote for group $i$ \\
    $\result$ & Election result \\
    $\filteradditions, \tallyproofdec$ & Full filter and tally proofs \\
    \bottomrule
  \end{tabular}
}
\end{table}

The TS uses standard zero-knowledge proofs of knowledge~\cite{Goldwasser:1985:KCI:22145.22178}
to prove that it operated correctly. We use the Fiat-Shamir heuristic~\cite{FiatS86} to convert them into non-interactive proofs of knowledge. We adopt the Camenisch-Stadler notation~\cite{Camenisch:1997:EGS:646762.706305}
to denote such proofs and write, for example,
\begin{equation*}
  \SPK\{ (\sk) : \pk = \generator^{\sk} \;\land\; \tokencounter = \ecdec(\sk, \encryptedindex) \}
\end{equation*}
to denote the non-interactive signature proof of knowledge that the prover knows
the private key $\sk$ corresponding to $\pk$ and that $\encryptedindex$ decrypts
to $\tokencounter$ under $\sk$.

\subsection{\name description}\label{sec:formal-desc}

\name proceeds in three phases: the pre-election phase, the election phase,
 and the tally phase. See Table~\ref{tab:notation} for a summary of frequently used symbols.

\subsubsection{Pre-election phase}
In the pre-election phase, the PBB publishes the candidates, and the TS and the trustees prepare 
their cryptographic material. The PA assigns a unique, random voter identifier $\vid_i$ to each 
eligible voter. The correspondence between voters and their identifiers is private to the PA. The PA 
also generates a random token index $\ballotnumber_i$ for each of the voters to enable the 
selection of the last ballot per voter. More formally: 

\begin{proc}[$\setup$]
  \label{proc:setup}
  To setup an election system with security parameter $\secpar$, electoral roll $\electoralroll$, candidate list $\candidatelist$, threshold $\threshold$, and $\lasttrustee$ trustees, the different entities run the 
  $\setup(1^{\secpar}, \electoralroll, \candidatelist, \threshold,
  \lasttrustee)$ procedure. First, they pick a group $\group$ with generator $\generator$ and prime order $\grouporder$. They then proceed with the following steps:
\begin{enumerate}
\item The PBB initializes the bulletin board, and adds the list of candidates $\candidatelist$ to the bulletin board.
\item The PA stores the electoral roll $\electoralroll$. Let $\nrvoters$ be the
  number of eligible voters on the electoral roll. The PA generates a random and
  unique voter
  identifier $\vid_i \in \group$ and ballot index $\ballotnumber_i \in
  \{2^{\secpar - 2}, \ldots, 2^{\secpar - 1} - 1\}$
  for each voter $\voter_i$ on the electoral
  roll and stores them internally.
  Finally, the PA generates a public-private key-pair $(\papk, \pask) = \signkeygen(1^\secpar)$ to sign tokens.
  It publishes $\papk$. 
\item The TS generates a public-private ElGamal key-pair $(\tspk, \tssk) = \eckeygen(\group, \generator, \grouporder)$. It publishes $\tspk$.
\item The trustees run $\votekeygen(1^{\secpar},\threshold, \lasttrustee, |\candidatelist|)$ to generate a public encryption key $\votepk$ and decryption keys $\voteski$ that the trustees keep private.
\end{enumerate}
\end{proc}

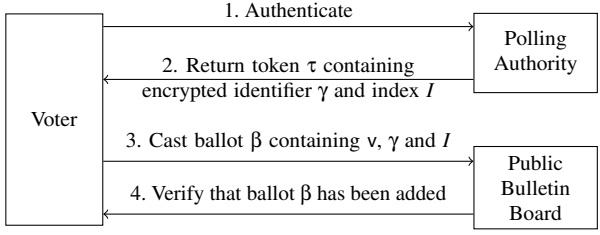
\begin{figure}[tbp]
  \centering
  \begin{tikzpicture}[
    party/.style={
      draw, rectangle,
      minimum height=3em,
      text centered,
      text width=4em,
      anchor=mid
    },
    ]
    \node[party,minimum height=8em,text width=3em] (v) {Voter};
    \node[party, above right =1em and 14em of v.east] (pa) {Polling\\ Authority};
    \node[party, below right =1em and 14em of v.east] (pbb) {Public\\Bulletin\\Board};

    \draw[->]
      ([yshift=1em] pa.west -| v.east) --
      node[above] {1. Authenticate}
      ([yshift=1em] pa.west);

    \draw[<-]
      ([yshift=-1em] pa.west -| v.east) --
      node[align=center, text width=13em] {2. Return token $\votetoken$ containing\\encrypted identifier $\encryptedvid$ and index $\encryptedindex$}
      ([yshift=-1em] pa.west);

    \draw[->]
      ([yshift=1em] pbb.west -| v.east) --
      node[above,text width=13em,align=center] {3. Cast ballot $\ballot$ containing $\encryptedvote$, $\encryptedvid$ and $\encryptedindex$}
      ([yshift=1em] pbb.west);

    \draw[<-]
      ([yshift=-1em] pbb.west -| v.east) --
      node[above] {4. Verify that ballot $\ballot$ has been added}
      ([yshift=-1em] pbb.west);
  \end{tikzpicture}
  \vspace{-4mm}
  \caption{\label{fig:overview-cast-vote} Election phase: Overview.}
\end{figure}

\subsubsection{Election phase}
In the election phase (see Figure~\ref{fig:overview-cast-vote}),
voters first authenticate to the PA to obtain an ephemeral voting token $\votetoken$.
They use this token to sign their ballot $\ballot$, % containing their encrypted vote, 
and post the ballot on the bulletin board. The bulletin board verifies
that the ballot is valid. We formalize this phase in three procedures:

\begin{proc}[$\gettoken(id,\text{Auth})$]
  \label{proc:gettoken}
  On input her identity $id$ and her inalienable means of
  authentication $\code{Auth}$:
  \begin{enumerate}
  \item The voter authenticates to the PA using $\code{Auth}$.
  \item The PA looks up the corresponding voter identifier $\vid_i$ and ballot
    index $\ballotnumber_i$. Then, the PA encrypts the voter
    identifier $\encryptedvid = \ecenc(\tspk, \vid_i)$ and ballot number
    $\encryptedindex = \ecenc(\tspk, \ballotnumber_i)$ (it first encodes
    $\ballotnumber_i$ as an element of $\group$), and increments the 
    ballot index $\ballotnumber_i:= \ballotnumber_i + 1$. The PA hides the index
    $\ballotnumber_i$ from the user to prevent coercers -- who can see what
    users can see under coercion -- from being able to detect whether the user revoted.
  \item The PA creates an ephemeral signing key $(\pk, \sk) = \signkeygen()$ and signs this key together with the encrypted voter identifier and ballot number:
    \begin{equation*}
     \tokensignature = \signsign(\pask, \pk \parallel \encryptedvid \parallel \encryptedindex)
    \end{equation*}
    and returns the token $\votetoken = (\pk, \sk, \encryptedvid, \encryptedindex, \tokensignature)$ to the user.
  \item \label{proc:gettoken:validate}
    The user verifies the token $\votetoken = (\pk, \sk, \encryptedvid, \encryptedindex, \tokensignature)$ by checking that $\signverify(\papk, \tokensignature, \pk \parallel \encryptedvid \parallel \encryptedindex) = \top$.
  \end{enumerate}
\end{proc}

\begin{proc}[$\vote(\votetoken, \candidate)$]
  \label{proc:vote}
	To cast a vote, the voter takes as private input the ephemeral voting token $\votetoken = (\pk, \sk, \encryptedvid, \encryptedindex, \tokensignature)$ and a candidate $c \in\candidatelist$, and then proceeds as follows:
  \begin{enumerate}
  \item \label{proc:vote:encrypt} Encrypts her candidate $c$ as
    $(\encryptedvote, \voteproof) = \voteenc(\votepk,\allowbreak c)$
    to obtain ciphertext $\encryptedvote$ and zero-knowledge proof of correct encryption $\voteproof$. 
  \item Creates the ballot
    \begin{equation*}
      \ballot = (\encryptedvote, \voteproof, \pk, \encryptedvid, \encryptedindex, \tokensignature, \signature)
    \end{equation*}
    where $\signature = \signsign(\sk, \encryptedvote \parallel \voteproof
    \parallel \pk \parallel \encryptedvid \parallel \encryptedindex \parallel
    \tokensignature)$. The voter posts the ballot $\ballot$ to the public bulletin board.
  \item The public bulletin board runs $\valid(\ballot)$, see below, to check
    that the ballot is valid, before appending it.
  \item Finally, the voter verifies that the ballot $\ballot$ has been appended to the bulletin board.
  \end{enumerate}
\end{proc}

\begin{proc}[$\valid(\ballot)$]
  \label{proc:valid}
  The bulletin board verifies that the ballot $\ballot = (\encryptedvote, \voteproof, \pk, \encryptedvid, \encryptedindex, \tokensignature, \signature)$ is valid with respect to the current state of the bulletin board as follows:
  \begin{enumerate}
  \item \label{proc:valid:checks}
    The PBB checks the correctness of the encrypted vote; of the user's signature using the ephemeral key $\pk$; and the PA's signature on this ephemeral key $\pk$, the encrypted voter identifier $\encryptedvid$, and the encrypted ballot number $\encryptedindex$:
    \begin{align*}
      & \voteverify(\votepk, \encryptedvote, \voteproof) = \top \\
      & \signverify(\pk, \signature, \encryptedvote \parallel \voteproof \parallel \pk \parallel \encryptedvid \parallel \encryptedindex \parallel \tokensignature) = \top \\
      & \signverify(\papk, \tokensignature, \pk \parallel \encryptedvid \parallel \encryptedindex) = \top.
    \end{align*}
  \item The PBB checks that neither the encrypted vote $\encryptedvote$ nor the key $\pk$ appear in any ballot $\ballot'$ on the bulletin board.
  \end{enumerate}
  If any of these checks fails, the bulletin board returns $\bot$, otherwise, the PBB returns $\top$.
\end{proc}

\begin{figure}[tbp]
  \centering
  \begin{tikzpicture}[
  party/.style={
      draw, rectangle,
      minimum height=3em,
      text centered,
      text width=5em,
      anchor=mid
    },
  ]
  \node[party, minimum height=8em, text width=3em] (pbb) {Public Bulletin Board};
  \node[party, above right =1em and 10em of pbb.east] (ts) {Tally Server};
  \node[party, below right =1em and 10em of pbb.east] (tt) {Trustees};

  \draw[->]
    ([yshift=1em] ts.west -| pbb.east)
    -- node[above] {1. Get ballots $\ballot_i$}
    ([yshift=1em] ts.west);
    
  \draw[<-]
    ([yshift=-1em] ts.west -| pbb.east) --
    node[above] {2. Post selected votes $\selectedvote{i}$}
    node[below] {and proof of correct filter $\filteradditions$}
    ([yshift=-1em] ts.west);

  \draw[->]
    ([yshift=1em] tt.west -| pbb.east)
    -- node[above] {3. Get $\ballot_i, \selectedvote{i}, \filteradditions$}
    ([yshift=1em] tt.west);

  \draw[<-]
    ([yshift=-1em] tt.west -| pbb.east) --
    node[above] {4. Post result $\result$ and}
    node[below] {proof of correct tally $\tallyproofdec$}
    ([yshift=-1em] tt.west);
  \end{tikzpicture}
  \vspace{-4mm}
  \caption{\label{fig:overview-tally} Tally phase: Overview}
\end{figure}
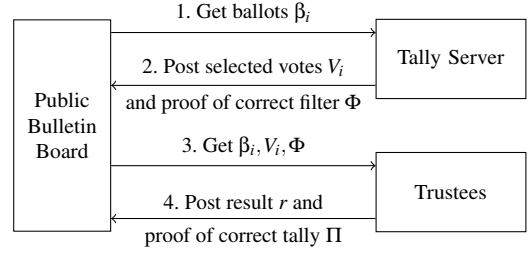

\begin{figure*}[tbp]
  \centering
  \begin{tikzpicture}[steparrow/.style={thick}]
    \node[align=center,draw] (pbb) {
      $\begin{array}{@{}c@{}}
        \ballot_1 \\
        \vdots \\
        \ballot_{\lastballot}
      \end{array}$
    };
    \node[below=0mm of pbb] {Ballots};

    \node[align=center,draw,right=1cm of pbb] (pbb-dummies) {
      $\begin{array}{@{}c@{}}
        \ballot_1 \\
        \vdots \\
        \ballot_{\lastballot} \\
        \ballot_{\lastballot + 1} \\
        \vdots \\
        \ballot_{n_T}
      \end{array}$
    };
    \node[below=0mm of pbb-dummies,align=center] {Ballots\\with dummies};

    \node[align=center,draw,right=of pbb-dummies] (shuffled) {
      \setlength\tabcolsep{3pt}
      \begin{tabular}{@{}cccc@{}}
      $\encryptedvote_1'$ & $\encryptedvid_1'$ & $\encryptedindex_1'$ & $\tstag_1'$ \\
      $\vdots$ & $\vdots$ & $\vdots$ & $\vdots$ \\
      $\encryptedvote_{n_T}'$ & $\encryptedvid_{n_T}'$ & $\encryptedindex_{n_T}'$ & $\tstag_{n_T}'$ \\
      \end{tabular}
    };
    \node[below=0mm of shuffled,align=center] {Shuffled ballots\\without proofs};

    \node[right=2cm of shuffled,align=center] (group-1) {
      $\vid_1$ \\[0mm]
      \begin{tabular}{|@{}ccc@{}|}
        \hline
        $\encryptedvote_{1,1}$ & $\tokencounter_1$ & $\tstag_{1,1}$ \\
        $\vdots$ & $\vdots$ & $\vdots$ \\
        $\encryptedvote_{1,\nrballotsingroup_1}$ & $\tokencounter_1 +
         \nrballotsingroup_1$ & $\tstag_{1,\nrballotsingroup_1}$ \\
        \hline
      \end{tabular}
    };
    \node[right=of group-1,align=center] (group-2) {
      $\vid_\kappa$ \\[0mm]
      \begin{tabular}{|ccc|}
        \hline
        $\encryptedvote_{\kappa,1}$ & $\tokencounter_\kappa$ & $\tstag_{\kappa,1}$ \\
        $\vdots$ & $\vdots$ & $\vdots$ \\
        $\encryptedvote_{\kappa,\nrballotsingroup_\kappa}$ & $\tokencounter_\kappa + 
        \nrballotsingroup_\kappa$ & $\tstag_{\kappa, \nrballotsingroup_\kappa}$ \\
        \hline
      \end{tabular}
    };
    \draw[-,dotted] (group-1) -- (group-2);
    \node[draw,fit=(group-1) (group-2)] (group) {};
    \node[above=0mm of group,align=center] {Grouped ballots};

    \node[below=1cm of group-1] (svote1) {$\preselectedvote{1}$};
    \node[below=1cm of group-2] (svote2) {$\preselectedvote{\kappa}$};
    \draw[->,gray] (group-1) -- (svote1);
    \draw[->,gray] (group-2) -- (svote2);
    \node[draw,fit=(svote1) (svote2)] (svotegroup) {};
    \draw[-,dotted] (svote1) -- (svote2);
    \node[below=0mm of svotegroup,align=center] {Selected votes including dummy voters};

    \node[left=5cm of svote1] (dvote1) {$\selectedvote{1}$};
    \node[right=1cm of dvote1] (dvote2) {$\selectedvote{\lastvoter}$};
    \draw[-,dotted] (dvote1) -- (dvote2);
    \node[draw,fit=(dvote1) (dvote2)] (dvotegroup) {};
    \node[below=0mm of dvotegroup,align=center] {Selected votes};

    \draw[->,steparrow] (pbb) -- node[above,align=center] {add\\dummies} (pbb-dummies);
    \draw[->,steparrow] (pbb-dummies) -- node[above,align=center] {shuffle} (shuffled);
    \draw[->,steparrow] (shuffled) -- node[above,align=center] {decrypt $\encryptedvid_i', \encryptedindex_i'$\\and group} (group);

    \draw[->,steparrow] (group) -- node[align=left,fill=white,fill opacity=0.8,text opacity=1] {Compute selected votes} (svotegroup.north -| group.south);
    \draw[->,steparrow] (svotegroup) -- node[above,align=center] {Shuffle and\\reveal+remove dummies} (dvotegroup);
  \end{tikzpicture}
  \vspace{-4mm}
 \caption{\label{fig:overview-filter-ballots} High-level overview of ballot
    filtering and grouping. Let $\lastballot$ be the number of ballots,
    $\lastdummy$ be the number of dummies, $n_T = \lastballot + \lastdummy$
    their sum, $\kappa$ be the number of voters plus number of dummy voters, and
    $\nrballotsingroup_i$ the number of (dummy) ballots for (dummy) voter $i$. First, the TS adds dummy
    ballots and proves they are well-formed.
    Then shuffles all ballots without the proofs, hiding
    which ballots were dummies. Then it verifiably decrypts both the encrypted
    voter identifiers $\encryptedvid_i'$ and the encrypted indices
    $\encryptedindex_i'$ to group the ballots by 
    $\vid$ and to select the last votes $\preselectedvote{i}$.
    Finally, it outputs the selected votes $\selectedvote{i}$ without dummies.
 }
\end{figure*}

\subsubsection{Tally phase}

In the tally phase (see Figure~\ref{fig:overview-tally}), the TS takes the ballots 
from the \pbb, adds dummy ballots, and shuffles them. Then, it selects the last
vote per voter (see
Figure~\ref{fig:overview-filter-ballots}). Then, to prevent dummy voters from making
an overhead in the shuffle and decrypt phase, it shuffles the selected ballots and 
removes all ballots cast by dummy voters.
Finally, the trustees shuffle and decrypt the selected ballots from real voters.
Formally, we define two procedures, one to filter votes ($\filter$), and one to tally the selected ballots ($\tally$):

\begin{proc}[$\filter$]
  \label{proc:filter}
  After the election closes, the TS selects the selected votes $\selectedvote{i}$ and produces the filter proof $\filteradditions$. If it aborts, it publishes the current $\filteradditions$ to the public bulletin board.

\begin{enumerate}
\item \label{proc:filter:checks}
  The tally server (TS) retrieves an ordered list of ballots $[\ballot_1, \ldots, \ballot_{\lastballot}]$ from the PBB, where
  $\ballot_i = (\encryptedvote_i, \voteproof_i, \pk_i, \encryptedvid_i, \encryptedindex_i, \tokensignature_i, \signature_i )$. The TS verifies the ballots by running step~\ref{proc:valid:checks} of \valid and verifies that there are 
  are no duplicate votes $\encryptedvote_i$ or ephemeral public keys $\pk_i$ on the bulletin board. If any of these checks fails, the TS sets $\filteradditions = \bot$, posts it to the bulletin board, and aborts.
\item \label{proc:add-dummies} The TS removes the proofs and signatures to
  obtain stripped ballots. It provably tags the ballots as `real' ballots using
  a deterministic ElGamal encryption (with randomness zero) of the value $g^0 = 1_{\group}$,  
  $\tstag_R = \ecenc(\tspk, \generator^0) = (g^0, g^0 \pk^{0}) = (1_{\group}, 1_{\group})$:
  \begin{equation*}
    \strippedballot_i = (\encryptedvote_i, \encryptedvid_i, \encryptedindex_i, \tstag_R).
  \end{equation*}
  Next, the TS creates $\lastdummy$ dummy ballots and provably tags them as such
  using a deterministic ElGamal encryption of the value $\generator$, $\tstag_D
  = \ecenc(\tspk, \generator) = (1_{\group}, \generator \cdot \pk^{0})$:
  \begin{equation*}
    \strippedballot_i = (\encryptedvote_{\epsilon}, \encryptedvid_i, \encryptedindex_i, \tstag_D),
  \end{equation*}
  where $i > \lastballot$ and $\encryptedvote_{\epsilon} = \votezeroenc(\votepk;
  0)$. We explain below how the TS determines the number of dummies $\lastdummy$
  as well as the values for $\encryptedvid_i$ and $\encryptedindex_i$.
  The TS adds the stripped ballots $\strippedballotlist = [\strippedballot_1, \ldots, \strippedballot_{\lastballot + \lastdummy}]$ to $\filteradditions$.
\item \label{proc:shuffle} The TS shuffles the stripped ballots $\strippedballotlist = [\strippedballot_1,\allowbreak \ldots,\allowbreak \strippedballot_{\lastballot + \lastdummy}]$ 
and randomizes the ciphertexts, to obtain a list of shuffled and randomized stripped ballots $\shuffledballotlist = [\shuffledballot_1, \ldots, \shuffledballot_{\lastballot + \lastdummy}]$, which it adds, together with a proof $\shuffleproof$ that this shuffle was performed correctly, to $\filteradditions$.
\item \label{proc:filter:decrypt-vid-index}
  The TS now operates on each shuffled ballot $\shuffledballot_i = (\encryptedvote_i', \encryptedvid_i', \encryptedindex_i', \tstag_i')$. It decrypts $\encryptedvid_i'$ to recover the shuffled and decrypted identifier, $\decryptedvid{i}$.
  It also decrypts $\encryptedindex_i'$ to obtain the shuffled ballot index $\decryptedindex{i}$ and proves it did so correctly: 
  \begin{align*}
    \decryptproof_i = \SPK\{&(\tssk): \tspk = \generator^{\tssk} \land \\
    &\decryptedvid{i} = \ecdec(\tssk, \encryptedvid_i') \land \\
    &\decryptedindex{i} = \ecdec(\tssk, \encryptedindex_i') \}
  \end{align*}
  It then adds 
  $\decryptedvidindexlist = [(\decryptedvid{1}, \decryptedindex{1}, \decryptproof_1),\allowbreak \ldots,\allowbreak (\decryptedvid{\lastballot+\lastdummy},\allowbreak \decryptedindex{\lastballot+\lastdummy}, \decryptproof_{\lastballot + \lastdummy})]$
  to $\filteradditions$. The TS aborts and adds $\bot$ to $\filteradditions$ if the decrypted ballot indices $\decryptedindex{i}$ are not unique for a given voter identifier. More precisely, it aborts if there exists indices $i,j; i\not=j$ such that $(\decryptedvid{i}, \decryptedindex{i}) = (\decryptedvid{j}, \decryptedindex{j})$.
\item\label{proc:filter:reencryption} The TS groups the ballots with the same voter identifier, and selects the
  ballot with the highest ballot index from each group. Let $\ballotgroup_1,
  \ldots, \ballotgroup_{\nrballotgroups}$ be the sets of ballot indices grouped
  by voter identifier. Consider group $\ballotgroup_j$ of size $\nrballotsingroup_j$.
 Let $j* = \argmax_{k,k \in  \ballotgroup_j} \decryptedindex{k}$ be the index for which the ballot
 index $\decryptedindex{j*}$ is maximal. Group $\ballotgroup_j$ either corresponds to a real voter, or to a fake voter. 
  The TS produces a reencryption $\preselectedvote{j}$ of the encrypted votes as follows:
  \begin{enumerate}
  \item If the group $\ballotgroup_j$ corresponds to a real voter, then the TS simply reencrypts the vote corresponding to the last ballot, i.e., it picks $r_j$ at random and sets
    \begin{equation*}
      \preselectedvote{j} = \encryptedvote_{j*} \cdot \votezeroenc(\votepk; r_j),
    \end{equation*}
    to a randomized encryption of $\encryptedvote_{j*}$.
  \item If the group $\ballotgroup_j$ corresponds to a fake voter, then picks $r_j$ at random and sets $\preselectedvote{j}$ to an empty vote:
    \begin{equation*}
      \preselectedvote{j} = \votezeroenc(\votepk; r_j).
    \end{equation*}
  \end{enumerate}
  The TS proves that it computed the $\preselectedvote{j}$ correctly. If the corresponding voter is real, then the ballot $\shuffledballot_{j*}$ selected in (a) should be a real ballot, so $\ecdec(\tssk, \tstag'_{j*})$ should equal $\generator^{0}$. If the voter is fake, then for all tags $\tstag'_{i_k}$ with $i_k\in\ballotgroup_j$, we have that $\ecdec(\tssk, \tstag'_{i_k}) = \generator^{1}$.
  Let $\ballotgroup_j = \{i_1, \ldots, i_{\nrballotsingroup_j}\}$ and 
  $\tstag = \prod_{k=1}^{\nrballotsingroup_j} \tstag'_{i_k}$,
  then the TS constructs the proof
  \begin{multline*}
    \reencryptionproof_j = \SPK\{ (r_j, \tssk) : \tspk = \generator^{\tssk} \land\\
      ((\generator^{0} = \ecdec(\tssk, \tstag'_{j*}) \land \preselectedvote{j} = \encryptedvote_{j*} \cdot \votezeroenc(\votepk; r_j)) \lor \\
      (\generator^{\nrballotsingroup_j} = \ecdec(\tssk, \tstag) \land \preselectedvote{j} = \votezeroenc(\votepk; r_j)))
    \}.
  \end{multline*}
  The TS adds the list of filtered encrypted votes $\filteredvotes = [(\vid_1, \preselectedvote{1}, \reencryptionproof_1), \ldots, (\vid_{\nrballotgroups}, \preselectedvote{\nrballotgroups}, \reencryptionproof_{\nrballotgroups})]$ to $\filteradditions$.
\item\label{proc:filter:second-shuffle} The list $\selectedvotelistdummies = [\preselectedvote{1},\allowbreak
\ldots,\allowbreak \preselectedvote{\nrballotgroups}]$ of selected votes 
contains ballots by dummy voters. In the next two steps, the TS removes
these. First, the TS shuffles
and randomizes the ciphertexts to obtain a new list $\selectedvotelistdummies' =
[\preselectedvote{1}',\allowbreak \ldots, \allowbreak
\preselectedvote{\nrballotgroups}']$ , which it adds, together with a proof
$\shuffleproof'$ of correct shuffle, to $\filteradditions$.
  \item The TS knows the indices $\dummyindices$ of votes in
    $\selectedvotelistdummies'$ that correspond to dummy voters and randomizers
    $r_i$ such that $\preselectedvote{i}' = \votezeroenc(\votepk; r_i)$ for $i
    \in \dummyindices$.
  The TS adds $\dummyindices$ and 
  $\openedrandomizers = [r_i]_{i\in\dummyindices} $ to $\filteradditions$.
  \item Finally, the TS publishes the remaining votes $\selectedvotelist = [\selectedvote{1},\allowbreak \ldots,\allowbreak \selectedvote{\lastvoter}]$ and the full proof $\filteradditions$ to the public bulletin board.
\end{enumerate}
\end{proc}

The filter procedure ensures that the TS cannot replace ballots by real voters: a selected 
vote must either correspond to a ballot by a real voter (condition a) or the selected 
vote is empty and the voter is a dummy voter (condition b). Moreover, the TS can only 
remove votes cast by dummy voters.

\begin{proc}[$\tally$]
  \label{proc:tally}
  To compute the final tally, the trustees proceed as follows:
\begin{enumerate}
\item The trustees verify that the TS operated honestly by running the $\verifyfilter()$ algorithm (see below). If \verifyfilter returns $\bot$ they return $(\result, \tallyproofdec) = (\bot, \bot)$.
\item\label{proc:tally:hom} Let
  $\selectedvotelist = [\selectedvote{1},\allowbreak \ldots,\allowbreak \selectedvote{\lastvoter}]$.
  The trustees jointly run the $(\result,
  \tallyproofdec) \gets \mixdecprotocol(\votepk,\allowbreak
  \selectedvotelist)$.
  They publish the election result $\result$ and the zero knowledge proof of
  correctness $\tallyproofdec$ to the public bulletin board.
\end{enumerate}
\end{proc}

\subsubsection{Verification}
Any external auditor can use the \pbb to verify that all
steps in the tally and filtering phases were performed correctly. We define the following verification procedures: 

\begin{proc}[$\verifyfilter$]
  \label{proc:verify-filter}
  Any party can verify that the filtering processes was performed correctly by running $\verifyfilter()$. This algorithm examines the content of the bulletin board and performs the following checks:
\begin{enumerate}
\item First, check if all ballots are correct and that no duplicate votes or public keys are included in the ballots as per step~\ref{proc:filter:checks} of \filter. If the checks fail, the bulletin board should contain $\filteradditions = \bot$; \verifyfilter returns $\bot$ if that is not the case. Otherwise, it continues.
\item It next retrieves the selected votes $\selectedvotelist$ and the proof $\filteradditions$ from the bulletin board and continues as follows:
  \begin{enumerate}
\item Verify that stripped real ballots are correctly formed.
    Consider ballots $[\ballot_1, \ldots, \ballot_{\lastballot}]$, where
    $\ballot_i = (\encryptedvote_i,\allowbreak \voteproof_i,\allowbreak \pk_i, \encryptedvid_i,
    \encryptedindex_i, \tokensignature_i, \signature_i )$ and check that the stripped ballot $\strippedballot_i = (\encryptedvote_i, \encryptedvid_i, \encryptedindex_i, \tstag_R)$ has been added to $\filteradditions$ (where $\tstag_R$ is as above).
\item \label{proc:verify-dummies} Verify that the dummy ballots on the
  bulletin board are correctly formed. For ballots
  $\strippedballot_{\lastballot + 1}, \ldots, \strippedballot_{\lastballot +
    \lastdummy}$ where $\strippedballot_i = (\encryptedvote_i, \encryptedvid_i,
  \encryptedindex_i, \tstag_i)$, check that $\encryptedvote_i =
  \encryptedvote_\epsilon$ and $\tstag_i = \tstag_D$ (where
  $\encryptedvote_\epsilon$ and $\tstag_D$ are as above). 
\item Let
    $\strippedballotlist = [\strippedballot_1, \ldots, \strippedballot_{\lastballot + \lastdummy}]$
    be all stripped ballots, and
    $\shuffledballotlist = [\shuffledballot_1, \ldots, \shuffledballot_{\lastballot + \lastdummy}]$
    the shuffled and randomized ballots.
    Verify the shuffle proof $\shuffleproof$
    to check that $\shuffledballotlist$ is a correct shuffle of $\strippedballotlist$.
\item Next, let
  $\decryptedvidindexlist = [(\decryptedvid{1}, \decryptedindex{1}, \decryptproof_1),\allowbreak \ldots,\allowbreak (\decryptedvid{\lastballot+\lastdummy},\allowbreak \decryptedindex{\lastballot+\lastdummy}, \decryptproof_{\lastballot + \lastdummy})]$ from the bulletin board, and verify the decryption proofs $\decryptproof_i$ for each of the shuffled ballots $\shuffledballot_i$.
\item
  Let $\vid_i'$ and $\ballotnumber_i'$ be the plaintexts verified in the previous step. Group the ballots by voter identifier into ballot groups $\ballotgroup_j$. For each group $\ballotgroup_j$, find ballot $\ballot_{j*}$ with the highest ballot index, recompute $\tau = \prod_{k = 1}^{\nrballotsingroup_j} \tau_{i_k}$, and verify the reencryption proof $\reencryptionproof_j$. 
\item Let $\selectedvotelistdummies$ be the selected votes $[\preselectedvote{1},\allowbreak
  \ldots,\allowbreak \preselectedvote{\nrballotgroups}]$ and 
  $\selectedvotelistdummies' = [\preselectedvote{1}',\allowbreak \ldots,\allowbreak \preselectedvote{\nrballotgroups}']$
  the shuffled and randomized votes. Verify the shuffle proof $\shuffleproof'$
  for $\selectedvotelistdummies$ and $\selectedvotelistdummies'$. 
\item Finally, for each $i \in \dummyindices$ verify that
  $\preselectedvote{i}' = \votezeroenc(\votepk; r_i)$ and that
  $\selectedvotelist =
  [\selectedvotelistdummies[i] \;|\; i \notin \dummyindices]$.
\end{enumerate}
If any of the checks fail, it returns $\bot$, and $\top$ otherwise.
\end{enumerate}
\end{proc}

\begin{proc}[\verify]
  \label{proc:verify}
  Any party can verify the result $\result$ and proof $\tallyproofdec$ against the public bulletin board. To do so, they proceed as follows:
  \begin{enumerate}
  \item Verify that the TS operated honestly by running the $\verifyfilter()$ algorithm. If \verifyfilter returns $\bot$, then return $\top$ if $(\result, \tallyproofdec) = (\bot, \bot)$, otherwise return $\bot$.
  \item Given the selected votes $\selectedvotelist$, return the result of
    $\voteverifytally(\votepk, \selectedvotelist, \result, \tallyproofdec).$
  \end{enumerate}
\end{proc}

% !TEX root = ../open_filtering.tex
\subsection{Hiding revoting patterns with dummies}
\label{sec:dummies}
In this section we provide a formal description of the dummy generation algorithm
introduced in Section~\ref{sec:key-ideas}. 

\para{Finding a cover.}
Formally, a cover is a set $\cover = \{(\groupsize_i, \nrvotersgroup_i)\}_i$ formed by groupings
$(\groupsize_i, \nrvotersgroup_i)\in\mathbb{Z}^+\times\mathbb{Z}^+$. Here, 
$\groupsize_i$ is the size of the ballot groups within that grouping, and
$\nrvotersgroup_i$ is the upper bound on the number of times
that such a ballot group can occur in any distribution of the $\lastballot$
real ballots among real voters.
We aim to find a cover of minimal size $\coversize = \sum_i
\groupsize_i \cdot \nrvotersgroup_i$ to minimize the
number of dummies added.

\parait{A sufficient cover.}
We derive an upper bound on the amount of dummies
required to build a cover. We do not use the number of real voters for this bound.
Let $\lastballot$ be the number of real ballots on the \pbb.
For simplicity, assume padded group sizes are powers of two, i.e., $\groupsize_i =
2^i$ for $i \geq 0$.
Given $\lastballot$ ballots, any distribution can have at most
$\nrvotersgroup_0 = \lastballot$ groups of size $\groupsize_0 = 1$ (one
ballot per voter).
Similarly, any distribution can have at most
$\nrvotersgroup_1 = \lfloor\lastballot/2\rfloor$ groups of size $\groupsize_1 = 2$.
Recall we pad ballot groups to the next bigger size, so a ballot group of 3
would be padded to one of size $\groupsize_2 = 4$ ballots, therefore
$\nrvotersgroup_2 = \lfloor\lastballot/3 \rfloor$. More generally,
there can be at most $\nrvotersgroup_i = \lfloor\lastballot/(2^{i-1} + 1)
\rfloor$ groups of $\groupsize_i = 2^i$ ballots. 
The biggest possible group (if all ballots were cast by the same voter), 
has size $2^{\lceil \log_2 \lastballot \rceil}$. 
Therefore, the size of the cover $\coversize$ is bounded by:
\begin{multline*}
  \coversize = \sum_{i = 0}^{\lceil \log_2 \lastballot \rceil} \nrvotersgroup_i
  \cdot \groupsize_i =
  \lastballot + 
  \sum_{i = 1}^{\lceil \log_2 \lastballot \rceil} 2^i \left\lfloor
    \frac{\lastballot}{2^{i - 1} + 1} \right\rfloor \\
  \leq
  \lastballot + 
  \sum_{i = 1}^{\lceil \log_2 \lastballot \rceil} \frac{2^i}{2^{i - 1} + 1} \lastballot
  \leq
  \lastballot + 
  \sum_{i = 1}^{\lceil \log_2 \lastballot \rceil} 2 \lastballot \\
  =
  (1 + 2\lceil \log_2 \lastballot \rceil) \lastballot.
\end{multline*}

\parait{An efficient cover.}
Knowing the number of real voters $\nrcastvoters$ enables to obtain a tighter cover.
Consider the example of Section~\ref{sec:key-ideas} with $\nrcastvoters = 2$ and $\lastballot = 9$.
If we only consider $\lastballot = 9$, one of the possible distributions
of votes would be having $\groupsize_1 = \lfloor 9 / 2 \rfloor = 4$ groups of size 2.
However, knowing $\nrcastvoters = 2$ rules out this possibility.
There can be at most one group of size two: if there were 2 groups, each of
the 2 voters could only cast 2 ballots, i.e., 4 ballots in total. However,
we know there are $9$ ballots so at least one voter has voted more than twice,
implying that $\groupsize_1 = 1$.

When the number of ballots grows this reasoning becomes intractable.
Consider ballot groups with group sizes, $\groupsize=\coversizebase^{i}$ for 
$i\in[0, \ldots, \lceil\log_\coversizebase \lastballot \rceil]$ for a real
number $\coversizebase > 1$.
We assume that $\lastballot > \nrcastvoters$, otherwise the cover would be trivial: $\cover=\{(\groupsize_0=1, \nrvotersgroup_0=\nrcastvoters)\}$. 
We compute the cover as follows.
\begin{enumerate}[noitemsep]
\item Consider groups of size $\groupsize_0 = k^0 = 1$. As $\lastballot >
  \nrcastvoters$,
  at least one voter must cast more than one ballot,
  resulting in $(\groupsize_0, \nrvotersgroup_0) = (1,
\nrcastvoters - 1)$.
\item Consider groups of size $\groupsize_i = \coversizebase^i$. We know that given
$\lastballot$, there can be at 
most $\alpha_i=\lfloor \lastballot /(\coversizebase^{i-1} + 1)\rfloor$ groups of size $\coversizebase^i$.
The number of groups is also bound by the number of voters. If $\nrcastvoters
\cdot \groupsize_i \geq \lastballot$ then all ballots can be assigned to the
$\nrcastvoters$ voters given groups of maximum size $\groupsize_i$, and we set $\nrcastvoters_i = \nrcastvoters$, otherwise
set $\nrcastvoters_i = \nrcastvoters - 1$ so that one voter is not in this grouping.
Finally, we need at least $\nrvotersgroup_i ( \coversizebase^{i - 1} + 1)$
ballots to make $\nrvotersgroup_i$ groups, but we must have enough ballots left
over to make $\nrcastvoters$ groups in total, i.e., $\lastballot \geq
\nrvotersgroup_i (\coversizebase^{i - 1} + 1) + (\nrcastvoters -
\nrvotersgroup_i).$ Rewriting gives bound
$\beta_i=\lfloor (\lastballot - \nrcastvoters) / \coversizebase^{i-1}\rfloor$.
We set $\nrvotersgroup_i = \min(\alpha_i, \nrcastvoters_i, \beta_i)$.
\end{enumerate}
Assuming $\lastballot > \nrcastvoters$,
the cover has $\coversize = \sum_{i =
  0}^{\lceil\log_\coversizebase \lastballot \rceil} \nrvotersgroup_i
\groupsize_i > \lastballot$ ballots, necessitating dummy ballots,
and $\sum_{i = 1}^{\lceil \log_\coversizebase \lastballot \rceil}
\nrvotersgroup_i > \nrcastvoters$ groups, necessitating dummy voters.

\para{Creating dummy voters and allocating dummy ballots.}\label{sec:dummy-votes}
The TS recovers all voter identifiers $\vid$ by decrypting the $\encryptedvid_i$s, 
and the corresponding ballot indices by decrypting the
$\encryptedindex_i$s.

So far, we assumed that ballot index sequences are continuous. However, there can
be gaps if some tokens were not used (e.g., the coercer does not use some tokens to 
identify index gaps in the filtering phase).
The TS first requests the number of obtained tokens $\lastballot'$
from the PA, and adds exactly $\lastballot' - \lastballot$ dummy ballots to fill
up any gaps, such  that $\lastballot'$ equals the
number of obtained tokens. The TS can create a dummy ballot for voter
$\vid$ by setting $\encryptedvid = \ecenc(\tspk, \vid)$.

Given the current number of ballots $\lastballot'$ and the number of real voters
$\nrcastvoters$ the TS computes a cover $\cover = \{(\groupsize_i, \nrvotersgroup_i)\}_i$.
To this end the TS performs a search to find the best $k$, i.e., the one
that gives the smaller cover. In our
experiments in Section~\ref{sec:evaluation}, $k$ tends to be in the $2$ to $4$
range, and the search takes less than a second.
The TS performs the following steps:

\begin{enumerate}[noitemsep]
\item For every voter $\vid_j, j \in \{1,\ldots,\nrcastvoters\}$ with $t$
  ballots, let
  $(\groupsize_i, \nrvotersgroup_i) \in\cover$ 
be the cover group with the smallest size $\groupsize_i$ such that $\groupsize_i
\geq t$. To ensure that dummy ballots are never counted,
the TS adds $t - \groupsize_i$ dummy votes to $\vid_j$ with descending
(and unused) ballot counters \emph{smaller} than
the last cast vote by this voter.  
\item For each grouping $(\groupsize_i, \nrvotersgroup_i) \in \cover$ let
  $\nrvotersgroup_i'$ be the number of real voters that were assigned to this
  group. The TS adds 
  $\nrvotersgroup_i - \nrvotersgroup_i'$ dummy voters. For each dummy voter, it
  picks a random $\vid$ and initial ballot index $\ballotnumber$ and creates
  $\groupsize_i$ dummy ballots with increasing ballot indices.
\end{enumerate}
In total, the TS adds $\lastdummy = \coversize - \lastballot$ dummies. Given
$\lastballot$ ballots on the bulletin board, $\lastballot + \lastdummy =
O(\lastballot \log \lastballot)$, see the upper bound above. As the filtering phases is linear
in $\lastballot + \lastdummy$ the time complexity of $O(\lastballot \log
\lastballot)$ follows.

% !TEX root = ../open_filtering.tex
\section{Security Analysis}\label{formal-sec}
We analyze \name's ballot privacy, verifiability, and coercion resistance. We follow Bernhard et al.~\cite{
Bernhard:2015:SCA:2867539.2867668} and model the trustees as a single trusted party with keys $(\votepk, \votesk)$, but
we note that the result holds when trustees are distributed. 
We explicitly model the bulletin board \pbb as an append only string $\BB$. To
ease modeling, we use the following 
redefinition of our voting scheme $\votingscheme = (\setup, \gettoken, \vote,
\valid, \filter, \verifyfilter, \tally, \verify)$ where the
algorithms output
changes to the bulletin board rather than posting to it directly.
While Bernhard et al. model voter registration implicitly,
we make the registration step explicit using the $\gettoken$
function because it forms an integral part of our voting scheme and may happen
more than once. The redefined algorithms in $\votingscheme$ are as follows:

\begin{itemize}
\item $\setup(1^\secpar, \electoralroll, \candidatelist)$ as in \setup in \cref{proc:setup} but explicitly returns the public key $\pk = (\papk, \tspk, \votepk)$
  and the corresponding private keys $\pask, \tssk, \votesk$.
\item $\gettoken(i)$ returns a token $\votetoken$ as in $\gettoken()$ in \cref{proc:gettoken}.
\item $\vote(\votetoken, \candidate)$ returns $\ballot$ as in $\vote(\votetoken, \candidate)$ in \cref{proc:vote} but does not post the ballot to the bulletin board. Moreover, the voter first verifies the token $\votetoken$ as in step~\ref{proc:gettoken:validate} of \cref{proc:gettoken}, and returns $\bot$ if it does not validate.
\item $\valid(\BB, \ballot)$ returns the result of $\valid(\ballot)$ in \cref{proc:valid} with respect to the bulletin board $\BB$.
\item $\filter(\BB, \lastballot', \tssk)$ as in $\filter$ in \cref{proc:filter},
  but takes the number of registrations $\lastballot'$ as explicit input, and 
  returns $\selectedvotelist \parallel \filteradditions$ instead of adding them to the board.
\item $\verifyfilter(\BB, \selectedvotelist, \filteradditions)$ runs
  $\verifyfilter$ from \cref{proc:verify-filter} on $\BB' = \BB \parallel
  \selectedvotelist \parallel \filteradditions$ and returns the result.
\item $\tally(\BB, \votesk)$ returns $(\result, \tallyproofdec)$ as in $\tally$ in \cref{proc:tally}.
\item $\verify(\BB, \result, \tallyproofdec)$ is as in $\verify$ in
  \cref{proc:verify} operating on the bulletin board $\BB \parallel \result \parallel \tallyproofdec$.
\end{itemize}

\subsection{Ballot privacy}
We base our ballot privacy definition on the game-based definition by Bernhard et al.~\cite{Bernhard:2015:SCA:2867539.2867668}.
They model ballot privacy using an indistinguishability game 
which simultaneously tracks two bulletin boards, $\BB_0$ for the ``real'' world
and $\BB_1$ for the ``fake'' world. Only one is accessible to the
adversary (see Figure~\ref{fig:bpriv-game}). 
The adversary, controlling the polling authority (PA) and the tally server (TS),needs to determine whether the tally was evaluated over the 
``real'' or ``fake'' world. It can decide how voters vote. 
Formally, the adversary can make calls to the oracle $\oraclevote(\votetoken, \candidate_0,
 \candidate_1)$
to let a user with token $\votetoken$ cast a vote for candidate $\candidate_0$
on $\BB_0$ and a vote for $\candidate_1$ on $\BB_1$; and to the oracle 
$\oraclecast(\ballot)$ to cast ballots $\ballot$
(constructed by the adversary) on $\BB_0$ and $\BB_1$. Because the adversary controls the PA, it can create as many voting tokens as it needs.

The outcome of the election is always computed on the real bulletin board
$\BB_0$. The adversary can once ask to compute the outcome by calling the oracle
$\oracletally(\selectedvotelist, \filteradditions)$ where $\selectedvotelist$, $\filteradditions$
is the output of $\filter$ computed by the adversary. The tally
oracle aborts if $\selectedvotelist, \filteradditions$ is not valid. If the adversary saw the ``real'' result corresponding to $\BB_0$, the tally protocol proceeds as normal and publishes a correct tally proof $\tallyproofdec$ with respect to $\BB_0$. If the adversary saw the ``fake'' bulletin board $\BB_1$, the experiment simulates the tally proof $\Pi$ with respect to $\BB_1$ using the algorithm $\simproof$ and returns the real result $r$.

\begin{figure}[tbp]
  \centering
  {\small
  \begin{tabular}{@{}ll@{}}
    \toprule
    \multicolumn{2}{@{}l}{$\bprivexp{\gamebit}{\adv}{\mathcal{V}}(\secpar, \electoralroll, \candidatelist)$:} \\
    & $(\pk, \pask, \tssk, \votesk) \gets \setup(1^{\secpar}, \electoralroll, \candidatelist)$ \\
    & $\gamebit \gets \adv^{\oracles}(\pk, \pask, \tssk)$ \\
    & Output $\gamebit'$ \\[1mm]
    \multicolumn{2}{@{}l}{$\oraclevote(\votetoken, \candidate_0, \candidate_1)$:} \\
    & Let $\ballot_0 = \vote(\votetoken, \candidate_0)$ and $\ballot_1 = \vote(\votetoken, \candidate_1)$ \\
    & If $\valid(\BB_{\gamebit}, \ballot_\gamebit) = \bot$ return $\bot$ \\
    & Else $\BB_0 \gets \BB_0 \parallel \ballot_0$ and
           $\BB_1 \gets \BB_1 \parallel \ballot_1$ \\[1mm]
    \multicolumn{2}{@{}l}{$\oraclecast(\ballot)$:} \\
    & If $\valid(\BB_{\gamebit}, \ballot) = \bot$ return $\bot$ \\
    & Else $\BB_0 \gets \BB_0 \parallel \ballot$ and
           $\BB_1 \gets \BB_1 \parallel \ballot$ \\[1mm]
    \multicolumn{2}{@{}l}{$\oracleboard()$:} \\
    & return $\BB_{\gamebit}$ \\[1mm]
    \multicolumn{2}{@{}l}{$\oracletally(\selectedvotelist, \filteradditions)$} \\
    & If $\verifyfilter(\BB_{\gamebit}, \selectedvotelist, \filteradditions) = \bot$ return $\bot$ \\
    & $\BB_{\gamebit} \gets \BB_{\gamebit} \parallel \selectedvotelist \parallel \filteradditions$ \\
    & $\BB_{1 - \gamebit} \gets \BB_{1 - \gamebit} \parallel \filter(\BB_{1 - \gamebit}, |\BB_{1 - \gamebit}|, \tssk)$ \\
    & $(\result, \tallyproofdec_0) \gets \tally(\BB_0, \votesk)$ \\
    & $\tallyproofdec_1 = \simproof(\BB_1, \result)$ \\
    & return $(\result, \tallyproofdec_{\gamebit})$ \\
    \bottomrule
  \end{tabular}
  }
  \vspace{-2mm}
  \caption{In the ballot privacy experiment
    $\bprivexp{\gamebit}{\adv}{\mathcal{V}}$, the adversary \adv has access to
    the oracles $\oracles = \{\oraclevote, \oraclecast, \oracleboard,
    \oracletally \}$.
    The adversary controls the TS and the PA.
    It can call $\oracletally$ only once.
  }
  \label{fig:bpriv-game}
\end{figure}

\begin{definition}
  \label{def:bpriv}
  Consider a voting scheme $\votingscheme = (\setup,\allowbreak \gettoken,\allowbreak \vote, \valid, \filter, \verifyfilter, \tally, \verify)$ for an electoral roll $\electoralroll$ and candidate list $\candidatelist$. We say the scheme has \emph{ballot privacy} if there exists an algorithm $\simproof$ such that for all probabilistic polynomial time adversaries \adv 
  \begin{multline*}
    \left|
      \textrm{Pr}\left[\bprivexp{0}{\adv}{\mathcal{V}}(\secpar, \electoralroll, \candidatelist) = 1 \right] - \right.
    \left.
      \textrm{Pr}\left[\bprivexp{1}{\adv}{\mathcal{V}}(\secpar, \electoralroll, \candidatelist) = 1 \right]
    \right|
  \end{multline*}
  is a negligible function in $\secpar$.
\end{definition}

%In the extended version of this paper~\cite{VoteAgainExtended}, we prove the
%following theorem.
In Appendix~\ref{app:proof-ballot-secrecy}, we prove the
following theorem.
\begin{theorem}
\label{thm:ballot-secrecy}
\name provides ballot privacy under the DDH assumption in the random oracle model.
\end{theorem}

Bernhard et al.\cite{Bernhard:2015:SCA:2867539.2867668} also define strong consistency, to ensure that the result $\result$ does not leak
information about individual ballots, and strong correctness to ensure that
valid ballots are never refused by the bulletin board. We restate these notions
and prove that \name satisfies them in 
%the extended version of this paper~\cite{VoteAgainExtended}.
Appendix~\ref{app:proof-ballot-secrecy}.

\subsection{Coercion resistance}
\newcommand{\crexp}[3]{\code{Exp}^{\code{cr},#1}_{#2,#3}}
\newcommand{\oraclegettoken}{\oraclesym\code{gettoken}}

Coercion resistance means that a coercer should not be able to determine whether a coerced user submitted to
coercion -- assuming it cannot learn this by seeing the result of the election 
(e.g., if there are zero votes for the selected candidate, the coercer knows the coerced user did not submit). 
In \name, this means that the coercer should not be able to determine whether a coerced user voted again, or not. 

\para{Existing coercion resistant models are insufficient.}
Juels, Catalano and Jakobsson (JCJ) model coercion resistance by comparing a
real-world game with an ideal game~\cite{Juels:2005:CEE:1102199.1102213}.
In JCJ, voters evade coercion by providing the coercer with a fake
credential. The real-world models normal execution. The adversary plays the
role of the coercer and chooses a set of corrupted voters and identifies the
coerced voter. Then, the honest voters cast their ballots (or abstain). 
If the coerced voter does not submit she
also casts her true ballot. Thereafter, the adversary is given the credentials
of all corrupt users, a credential for the coerced voter (which is fake if that
voter resists), and the current
bulletin board.
The adversary can now cast more ballots. Upon seeing the
result and the tally proof the adversary decides if the coerced voter submitted.
In the ideal game, the adversary is not shown the content of the bulletin
board, and she is given the true credential of the coerced
voter and can therefore cast real ballots for the coerced voter. However, a
modified tally function does not count ballots for the coerced voter cast by the
adversary if the coerced voter resists. Once the election phase is over, the
adversary is shown only the tally result, not the tally proof.

The JCJ model does not work for the revoting setting where the coerced voter casts another ballot 
\emph{after} casting the ballot under coercion.
Achenbach et al.~\cite{193462} propose a variant 
in which the coerced voter acts after the adversary has cast his votes,
revoting if she resists or doing nothing if she submits.
Thereafter, the adversary is shown 
the new bulletin board and the resulting tally and proof. In the ideal model, the adversary is only provided the length of the bulletin board.

\begin{figure*}[tbp]
 \centering
 \begin{tabular}{@{}l@{}}
 \toprule
 \begin{minipage}[b][][t]{0.9\textwidth}
  \small \centering
  \begin{tabular}{@{}ll@{}}
    \multicolumn{2}{@{}l}{$\crexp{\gamebit}{\adv}{\mathcal{V}}(\secpar, \electoralroll, \candidatelist)$:} \\
    & $(\pk, \pask, \tssk, \votesk) \gets \setup(1^{\secpar}, \electoralroll, \candidatelist)$ \\
    & Create $\PA_0$ and $\PA_1$ with keys $\papk^0, \papk^1$ \\
    & $\gamebit' \gets \adv^{\oracles}(\pk, \papk^\gamebit, \papk^{1 - \gamebit})$ \\
    & Output $\gamebit'$ \\[1mm]
    \multicolumn{2}{@{}l}{$\oraclevote(i_0, \candidate_0, i_1, \candidate_1)$:} \\
    & Let $\votetoken_0 \gets \PA_0.\gettoken(i_0)$ and
          $\votetoken_1 \gets \PA_1.\gettoken(i_1)$ \\
    & Let $\ballot_0 = \vote(\votetoken_0, \candidate_0)$ and $\ballot_1 = \vote(\votetoken_1, \candidate_1)$ \\
    & If $\valid(\BB_{\gamebit}, \ballot_\gamebit) = \bot$ return $\bot$ \\
    & Else $\BB_0 \gets \BB_0 \parallel \ballot_0$ and
           $\BB_1 \gets \BB_1 \parallel \ballot_1$ \\[1mm]
    \multicolumn{2}{@{}l}{$\oraclegettoken(i)$:} \\
    & Let $\votetoken_0 \gets \PA_0.\gettoken(i)$ and
          $\votetoken_1 \gets \PA_1.\gettoken(i)$ \\
    & return $\votetoken = \votetoken_\gamebit, \votetoken' = \votetoken_{1-\gamebit}$ \\[2mm]
  \end{tabular}
  \phantom{space}
  \begin{tabular}{@{}ll@{}}
    \multicolumn{2}{@{}l}{$\oraclecast(\ballot, \ballot')$:} \\
    & Let $\ballot_{\gamebit} \gets \ballot$ and $\ballot_{1 - \gamebit} \gets \ballot'$ \\
    & If $\valid(\BB_{0}, \ballot_{0}) = \bot$ or
         $\valid(\BB_{1}, \ballot_{1}) = \bot$ return $\bot$ \\
    & Else $\BB_0 \gets \BB_0 \parallel \ballot_0$ and
           $\BB_1 \gets \BB_1 \parallel \ballot_1$ \\[1mm]
    \multicolumn{2}{@{}l}{$\oracleboard()$:} \\
    & return $\BB_{\gamebit}$ \\[1mm]
    \multicolumn{2}{@{}l}{$\oracletally()$} \\
    & Let $\lastballot'$ be the number of tokens obtained from $\PA_0$. \\
    & Let $\selectedvotelist_0, \filteradditions_{0} = \filter(\BB_{0}, \lastballot', \tssk)$
    \\
    & Let $(\result, \tallyproofdec_0) \gets \tally(\BB_0 \parallel \selectedvotelist_0 \parallel \filteradditions_0, \votesk)$ \\
    & Let $(\selectedvotelist_1, \filteradditions_1) \gets \simfilter(\BB_1, \lastballot', \result)$ \\
    & Let $\BB_{0} \gets \BB_{0} \parallel \selectedvotelist_0 \parallel \filteradditions_0$ and
    $\BB_{1} \gets \BB_{1} \parallel \selectedvotelist_1 \parallel \filteradditions_1$ \\
    & $\tallyproofdec_1 = \simproof(\BB_1, \result)$ \\
    & return $(\result, \tallyproofdec_{\gamebit})$ \\
  \end{tabular}
  \end{minipage}\\[-1mm]
  \bottomrule
  \end{tabular}
  \vspace{-3mm}
  \caption{In the coercion resistance experiment
    $\crexp{\gamebit}{\adv}{\mathcal{V}}$, adversary \adv has access to
    oracles $\oracles = \{\oraclevote,\allowbreak \oraclegettoken,\allowbreak \oraclecast,\allowbreak \oracleboard,\allowbreak \oracletally \}$.
    It can call $\oracletally$ only once, thereafter it can see the result $\filteradditions_{\gamebit}$ by using $\oracleboard()$.}
  \label{fig:coercion-resistance-game}
\end{figure*}

The model proposed by Achenbach et al.~\cite{193462} does not capture coercion resistance. 
Following the real/ideal paradigm,
in the ideal game
it should hold with overwhelming probability that the adversary cannot distinguish between a 
submitting and a resisting coerced voter.
Then, the proof would show that the adversary cannot learn more in the
real world than it could in the ideal world. 
However, in the ideal game proposed by Achenbach et al., the coercion resistance property does not hold. 
The adversary can \emph{always} distinguish between these two cases by simply observing the length of the bulletin board 
(which increases by one ballot if the coerced voter revotes). Therefore, any proofs in this model say nothing 
about whether the real scheme offers coercion resistance. While the Achenbach et
al.~\cite{193462} scheme seems to be coercion resistant,
coercion resistance does not follow from the proof in their model.

Finally, the model by Achenbach et al.\ does not capture the leakage resulting from the state kept by the voter, 
or as in our protocol, by the polling authority. Our protocol deliberately hides the ballot counter from the voter, 
so that when if the coercer coerces the voter again, it cannot determine whether the coerced 
voter re-voted based on this counter. In Achenbach et al.'s model, the coercer
cannot coerce a voter more than once.

\para{A new coercion resistance definition.}
We propose a new game-based coercion resistance definition
inspired by Bernhard et al.'s ballot privacy definition.
The game tracks two bulletin boards, $\BB_0$ and $\BB_1$, of which only one is accessible to the
adversary (depending on the bit $\gamebit$). We ensure that regardless of the bit $\gamebit$, the same 
number of ballots are added to the bulletin board. The goal of the adversary is
to determine $\gamebit$
(see Figure~\ref{fig:coercion-resistance-game}).
Recall that we assume that the PA, TS, and trustees are honest with respect to coercion resistance.

To model submits versus resists, we provide the adversary with an
$\oraclevote(i_0, \candidate_0, i_1, \candidate_1)$ oracle to let voter $i_0$, a
``coerced'' voter, cast a vote for candidate $\candidate_0$ in $\BB_0$, and
voter $i_1$, any other voter, cast a vote for candidate $\candidate_1$ in
$\BB_1$. The adversary is allowed to make this call multiple times.
Regardless of the value of $\gamebit$, every call to $\oraclevote$ results in a
single ballot being added to each BB. This prevents the trivial win in the Achenbach et al.\ model. 
Since the polling authority keeps state, we work with two PAs: $\PA_0$ and
$\PA_1$.

We model a coercion attack as follows. The adversary can cast votes using any
user by calling $\oraclegettoken(i)$ to obtain a voting token
$\votetoken$ for voter $i$ on the board that it can see, and a token
$\votetoken'$ for the other board. 
It can then run $\ballot = \vote(\votetoken, \candidate)$ and $\ballot' =
\vote(\votetoken', \candidate)$ itself to create ballots for candidate
$\candidate$, on both boards and cast them
using $\oraclecast(\ballot, \ballot')$. Note that per our assumptions, the
adversary does not get access to the voter's means of authentication. Moreover,
we require that the adversary always casts valid ballots to both boards (but the
encoded candidate need not be the same).

Finally, the adversary can make one call to $\oracletally()$ which performs the filtering step and returns the result $\result$ 
(always computed on $\BB_0$)
and the tally proof.
The result of $\filter$ is accessible using $\oracleboard$.
To correct for leakage stemming from the tally result, as in the ballot
privacy game, we simulate the filter and tally proofs if the adversary sees $\BB_1$.

This game models all the coercion attacks applicable to \name:
\begin{itemize}
\item \emph{The 1009 attack}. The
  adversary casts a ballot as coerced voter $i_0$ using $\votetoken, \votetoken' =
  \oraclegettoken(i_0)$, $\ballot = \vote(\votetoken, \candidate)$, $\ballot' =
  \vote(\votetoken', \candidate)$ and then
  $\oraclecast(\ballot, \ballot')$ 1009 times. (Both boards now contain 1009
  ballots by voter $i_0$.) Then it calls $\oraclevote(i_0, \candidate,
  i_1, \candidate)$. If $\gamebit = 0$ the coerced voter revotes for candidate
  $\candidate$ on $\BB_0$, otherwise it does not, and the alternative voter casts a ballot
  for candidate $\candidate$ on $\BB_1 $ visible to the adversary.
  Note that if the result of $\filter$ $\filteradditions$ reveals the size of
  a group of ballots, the adversary can win this game 
  (\simfilter does not model this leakage as it only gets 
  $\lastballot'$ and $\result$ as input).
\item \emph{Returning coercer}.
  Let voter $i_0$ be the coerced voter. First
  the coercer runs $\votetoken, \votetoken' = \gettoken(i_0)$, $\ballot =
  \vote(\votetoken, \candidate)$ and $\ballot' = \vote(\votetoken',
  \candidate)$, and $\oraclecast(\ballot, \ballot')$ to cast one vote as the
  coerced user on both boards and to observe the token $\votetoken$
  corresponding to the board $\BB_{\gamebit}$ it can see. Then it
  runs $\oraclevote(i_0, \candidate_0, i_1, \candidate_1)$, causing $i_0$ to
  cast a vote on the bulletin board $\BB_{\gamebit}$ if $\gamebit = 0$,
  and $i_1$ to casts a vote on $\BB_{\gamebit}$ if $\gamebit = 1$.
  Thereafter, it can examine the state by running
  $\votetoken, \votetoken' = \gettoken(i_0)$ again. If the new token $\votetoken$
  leaks whether voter $i_0$ voted again (on board $\BB_{\gamebit}$), then the adversary 
  wins the coercion resistance game.
\end{itemize}

\begin{definition}
  \label{def:cr}
  Consider a voting scheme $\votingscheme = (\setup,\allowbreak \gettoken,\allowbreak \vote, \filter,
  \verifyfilter, \tally, \verify)$ for an electoral roll $\electoralroll$ and
  candidate list $\candidatelist$. We say the scheme has \emph{coercion
    resistance} if there exist algorithms $\simfilter$ and $\simproof$ such that for all probabilistic polynomial time adversaries \adv
  \begin{equation*}
    \left|
      \textrm{Pr}\left[\crexp{0}{\adv}{\mathcal{V}}(\secpar, \electoralroll, \candidatelist) = 1 \right] -
      \textrm{Pr}\left[\crexp{1}{\adv}{\mathcal{V}}(\secpar, \electoralroll, \candidatelist) = 1 \right]
    \right|
  \end{equation*}
  is a negligible function in $\secpar$.
\end{definition}

In Appendix~\ref{app:proof-coercion-resistance}, we prove the
following theorem.
\begin{theorem}
\label{thm:coercion-resistance}
\name provides coercion resistance under the DDH assumption in the random oracle model.
\end{theorem}

\subsection{Verifiability}
In their analysis, Achenbach et al.~\cite{193462} adapt the correctness definition
of Juels et al.~\cite{Juels:2005:CEE:1102199.1102213} to the revoting
setting. However, Achenbach et al.'s model does not take into account that
voters may not check that their ballots are cast correctly, nor that newer ballots 
should supersede older ballots \emph{even} if voters have been coerced or corrupted.
To address these cases, we adapt the qualitative game-based verifiability
definition of Cortier et al.~\cite{CortierGGI14} -- which accounts for a malicious
bulletin board and voters not checking their ballots -- to our setting by adding
the \gettoken function and explicitly modeling revoting.
As in Cortier et al.~\cite{CortierGGI14}, our game does not model voter's intent, 
and assumes that the voting hardware, i.e., the device and software running 
\vote, is honest. We refer to Cortier et al.~\cite{CortierGKMT16} for a formal process-based 
computational model that does model verifiability with voter intent.
We note that the correctness 
definition by Juels et al.~\cite{Juels:2005:CEE:1102199.1102213} was renamed to 'verifiability' by Cortier et al.~\cite{CortierGGI14}, and 
therefore any model satisfying the latter satisfies the former.

\begin{figure*}[!htb]
 \centering
  {\small
  \begin{minipage}[t]{0.72\textwidth}
  \vspace{0cm}
  \begin{tabular}{@{}>{\tiny}rl@{}l@{}}
    \toprule
    0 & \multicolumn{2}{@{}l}{$\verifbexp{\gamebit}{\adv}{\mathcal{V}}(\secpar, \electoralroll, \candidatelist)$:} \\
    1 & & $(\pk, \pask, \tssk, \votesk) \gets \setup(1^{\secpar}, \electoralroll, \candidatelist)$ \\
    2 & & Set $\honestvotes \gets \emptyset$ and $\corruptedset \gets \emptyset$ \\
    3 & & $(\BB, \selectedvotelist, \filteradditions, \result, \tallyproofdec) \gets \adv^{\oracles}(\pk, \tssk, \votesk)$ \\
    4 & & If $\verifyfilter(\BB, \selectedvotelist, \filteradditions) = \bot$
          or $\verify(\BB \parallel \selectedvotelist \parallel \filteradditions, \result, \tallyproofdec) = \bot$ return $0$ \\
    5 & & Let $\checked = \{ (i_1, \ctr_1), \ldots, (i_{\nrchecks}, \ctr_{\nrchecks}) \}$ correspond to checked ballots. \\
    6 & & Let $\textsf{Corrupted} =
      \{ i \;|\; (i, \ctr) \in \corruptionevents \land
             \forall (i, \ctr') \in \checked: \ctr' < \ctr \}$ \\
    7 & & Let $\textsf{Checked} =
      % \{i_1^{V}, \ldots, i_{\nrverified}^{V} \} = 
      \{ i \;|\; (i, \_) \in \checked \} \setminus \textsf{Corrupted}$\\
    8 & & Let $\textsf{Unchecked}
      % = \{i_1^U, \ldots, i_{\nrunchecked}^U\}
      = \{ i \;|\; (i, \_, \_) \in \honestvotes \land (i, \_) \not\in \corruptionevents \} \setminus \textsf{Checked}$ \\
    9 & & Let $\textsf{AllowedVotes}[i] = \{ \candidate \;|\; (i, \ctr, \candidate) \in \honestvotes \textrm{ s.t. } \forall (i, \ctr') \in \checked: \ctr \geq \ctr' \}$ \\
    10 & & If $\exists \;
      \candidate_1^{V}, \ldots, \candidate_{\nrverified}^{V} \textrm{ s.t. } \candidate_j \in \textsf{AllowedVotes}[i_j^{V}]$ where $\textsf{Checked} = \{i_1^V, \ldots, i_{\nrverified}^V\}$ 
      \\
    11 & & \phantom{If} $\exists \;
      (i_1^{U}, \candidate_1^{U}), \ldots,
      (i_{\nrunchecked}, \candidate_{\nrunchecked}^U) \textrm{ s.t. }
      i_j^{U} \in \textsf{Unchecked}, 
      \candidate_j^{U} \in \textsf{AllowedVotes}[i_j^{U}]$,
      $i_j^{U}$ different  \\
    12 & & \phantom{If} $\exists \; \candidate_1^{B}, \ldots, \candidate_{\nrbad}^{B}
      \in \candidatelist$ s.t. $0 \leq \nrbad \leq |\textsf{Corrupted}|$\\
    13 & & \phantom{If} s.t. $\result =
      \partialtally(\{\candidate_i^{V}\}_{i = 1}^{\nrverified}) \star_R
      \partialtally(\{\candidate_i^{U}\}_{i = 1}^{\nrunchecked}) \star_R
      \partialtally(\{\candidate_i^{B}\}_{i = 1}^{\nrbad})$ \\
    14 & & Then return $0$, otherwise return $1$ \\
    \bottomrule
  \end{tabular}
  \end{minipage}
  \begin{minipage}[t]{0.27\textwidth}
    \vspace{0cm}
  \fbox{
  \begin{tabular}{@{}ll@{}}
    \multicolumn{2}{@{}l}{$\oraclevotesimple(i, \candidate)$:} \\
    & Let $\votetoken = \gettoken(i)$ \\
    & Add $(i, \nrtokens{i}, \candidate)$ to $\honestvotes$ \\
    & Return $\vote(\votetoken, \candidate)$ \\[2mm]
    \multicolumn{2}{@{}l}{$\oraclegettoken(i)$:} \\
    & Let $\votetoken = \gettoken(i)$ \\
    & Add $(i, \nrtokens{i})$ to $\corruptionevents$ \\
    & return $\votetoken$ \\
  \end{tabular}}
  \end{minipage}
  }
  \vspace{-3mm}
  \caption{In the verifiability game experiment
    $\verifbexp{\gamebit}{\adv}{\mathcal{V}}$, the adversary \adv has access to
    the oracles $\oracles = \{\oraclegettoken, \oraclevotesimple\}$.}
  \label{fig:verb-game}
\end{figure*}

In a nutshell, a voting scheme is verifiable~\cite{CortierGGI14} if for 
$\nrcorrupted$ corrupt voters, the result of the election always includes: 
(1) all votes by honest voters that verified whether their ballots were cast correctly, 
(2) at most $\nrcorrupted$ corrupted votes, and
(3) a subset of the votes by honest voters that did not check if their ballots were cast correctly. 
These conditions ensure that while a malicious bulletin board can drop ballots of voters 
that do not check, it can insert at most $\nrcorrupted$ new votes. 

\para{Extending the current verifiability definition.} We extend the definition
presented by Cortier et al.~\cite{CortierGGI14} for the revoting setting to
explicitly consider the number of votes cast by a voter, see Figure~\ref{fig:verb-game}.
The PA is honest, but the adversary controls the bulletin board, the TS, and the trustees.
The system implicitly tracks the number of tokens $\nrtokens{i}$ that have been obtained by voter $i$.
The game tracks when each voter is corrupted in a (initially empty) list of corruption events $\corruptionevents$, 
and tracks the honest votes in $\honestvotes$. The adversary can call two oracles: $\oraclevotesimple(i, \candidate)$ to request 
that honest voter $i$ outputs a ballot for candidate $\candidate$, and $\oraclegettoken(i)$ to get a voting 
token for user $i$. Note that this models \emph{both corruption and coercion} of voter $i$. After a call to $\oraclegettoken(i)$, 
voter $i$ is considered corrupted until it casts an honest ballot using $\oraclevotesimple(i, \candidate)$.
Eventually, the adversary outputs a bulletin board $\BB$, the
selected votes $\selectedvotelist$ and proof $\filteradditions$,
the election outcome $\result \in \resultspace$, and a tally proof
$\tallyproofdec$ (line 3). The adversary loses if
$\filteradditions$ or $\tallyproofdec$ do not verify (line 4). If it verifies,
the adversary wins if the result does not 
satisfy the three intuitive conditions above.  

The game computes the following groups of voters:
\begin{itemize}
  \item \textsf{Corrupted} (line 6): voters considered corrupted,
    i.e., voters that were once corrupted (by calling $\oraclegettoken$)
    and thereafter never cast a checked honest vote.
  \item \textsf{Checked} (line 7): voters that verified a ballot and were
    not corrupted thereafter.
  \item \textsf{Unchecked} (line 8): voters that were never corrupted, but did not
    check their ballots either.
\end{itemize}
The game computes allowed candidates for honest voters:
\begin{itemize}
  \item $\textsf{AllowedVotes}[i]$ (line 9) A list of candidates that voter $i$ honestly voted for in or after the last checked ballot. If voter $i$ never checked a ballot, this list includes all candidates this voter ever voted for.
\end{itemize}

The adversary wins if the result $\result$ verifies but 
violates any of the following conditions (lines 10--13):
(1) For each honest voter that verified a ballot and was not thereafter
corrupted (i.e., voters in \textsf{Checked}) the result should include either the candidate in that ballot, or a candidate in a later ballot. This corresponds to the candidates $\{\candidate_i^{V}\}_{i = 1}^{\nrverified}$ in the game.
(2) Of the honest voters that did not check their ballots but were never
corrupted (i.e., voters in \textsf{Unchecked}), at most one candidate that the honest voter voted for (in any ballot) can be included. This corresponds to the candidates
  $\{\candidate_i^{U}\}_{i = 1}^{\nrunchecked}$ in the game, where $\nrunchecked$ can be smaller than $|\textsf{Unchecked}|$ or in fact 0.
(3) At most $\nrcorrupted$ corrupted (or bad) votes were counted (i.e., the candidates $\{\candidate_i^{B}\}_{i = 1}^{\nrbad}$)

In the game, the sum of these choices is modeled by the tallying
function $\partialtally:
\candidatelist^* \to \resultspace$ that maps the voter's choices in 
$\candidatelist$ to an election result in $\resultspace$. This function should
support partial tallying, i.e., for any two lists $S_1$ and $S_2$ we have
that $\partialtally(S_1 \cup S_2) = \partialtally(S_1) \partialoperator
\partialtally(S_2)$ for a commutative binary operator $\partialoperator :
\resultspace \times \resultspace \to \resultspace$.
Note that a tally function that outputs the number of votes per candidate 
naturally admits partial tallying.
\begin{definition}
  \label{def:ver}
  Consider a voting scheme $\votingscheme = (\setup,\allowbreak \gettoken,\allowbreak \vote, \filter,
  \verifyfilter, \tally, \verify)$ for an electoral roll $\electoralroll$ and
  candidate list $\candidatelist$. We say the scheme is \emph{verifiable}
  if for all probabilistic polynomial time adversary \adv
  \begin{equation*}
    \left|
      \textrm{Pr}\left[\verifbexp{0}{\adv}{\mathcal{V}}(\secpar, \electoralroll, \candidatelist) = 1 \right] -
      \textrm{Pr}\left[\verifbexp{1}{\adv}{\mathcal{V}}(\secpar, \electoralroll, \candidatelist) = 1 \right]
    \right|
  \end{equation*}
  is a negligible function in $\secpar$.
\end{definition}

%In the extended version of this paper~\cite{VoteAgainExtended}, we prove the
%following theorem.
In Appendix~\ref{app:proof-verifiability}, we prove the
following theorem.
\begin{theorem}\label{thm:verifiability}
\name is verifiable under the DDH assumption in the random oracle model.
\end{theorem}

% !TEX root = ../open_filtering.tex

\section{Performance Evaluation}
\label{sec:evaluation}

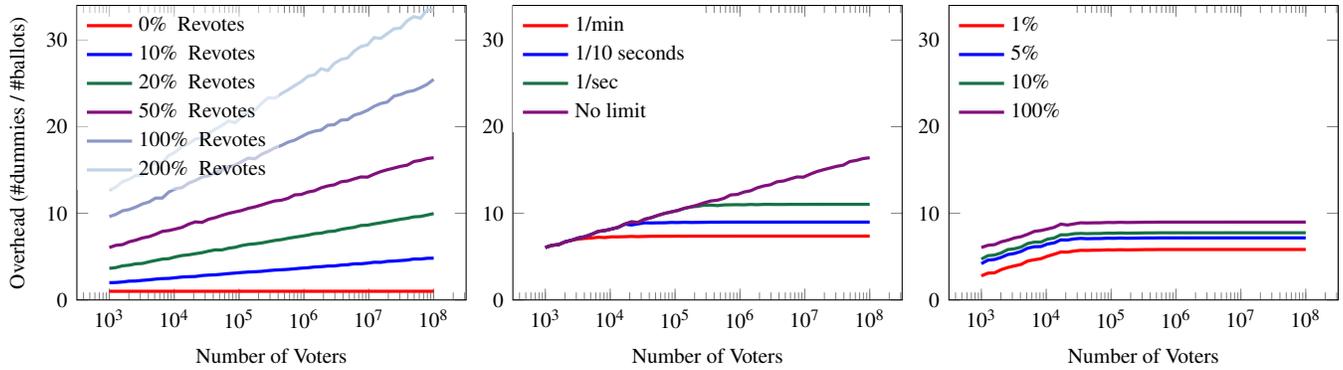
\begin{figure*}[tb]
  \begin{tikzpicture}\label{overhead-unlimited}
    \begin{semilogxaxis}[
      xlabel={Number of Voters},
      ylabel={Overhead (\#dummies / \#ballots)},
      ymin=0, ymax=34,
      width=0.38\textwidth,
      height=5.5cm,
      legend style={fill opacity=.5,text opacity=1},
    ]
    \foreach \p in {0, 10, 20, 50, 100, 200}{
      \addplot table [mark=none,green!20!black,x=NrVoters,y=\p,col sep=comma]{parts/csv-files/normal_overhead.csv};
      \addlegendentryexpanded{\p\%  \text{ Revotes}}
    }
    \end{semilogxaxis}
  \end{tikzpicture}
  \begin{tikzpicture}\label{overhead-1ps}
    \begin{semilogxaxis}[
      xlabel={Number of Voters},
      ymin=0, ymax=34,
      width=0.38\textwidth,
      height=5.5cm,
    ]
    \foreach \p in {1/min,1/10 seconds,1/sec,No limit}{
      \addplot table [mark=none,green!20!black,x=NrVoters,y=\p,col sep=comma]{parts/csv-files/max_votes_limit.csv};
      \addlegendentryexpanded{\p}
    }
    \end{semilogxaxis}
  \end{tikzpicture}
  \begin{tikzpicture}\label{overhead-majority-one}
	\begin{semilogxaxis}[
    xlabel={Number of Voters},
    ymin=0, ymax=34,
    width=0.38\textwidth,
    height=5.5cm,
	]
	\foreach \p in {1,5,10,100}{
		\addplot table [mark=none,green!20!black,x=NrVoters,y=\p,col sep=comma]{parts/csv-files/max_votes_voter_limit.csv};
		\addlegendentryexpanded{\p\%}
	}
	\end{semilogxaxis}
  \end{tikzpicture}
  \vspace{-8mm}
  \caption{\label{fig:dummies-overhead} Dummy ballots overhead:
  Varying percentages of revotes (left); limiting to 50\% overhead (center);
  and limiting to 50\% overhead, 1 ballot per 10
    seconds, and bounding the percentage of voters revoting (right).}
\end{figure*}  

\begin{figure*}[tb]
  \includegraphics[width=0.33\textwidth]{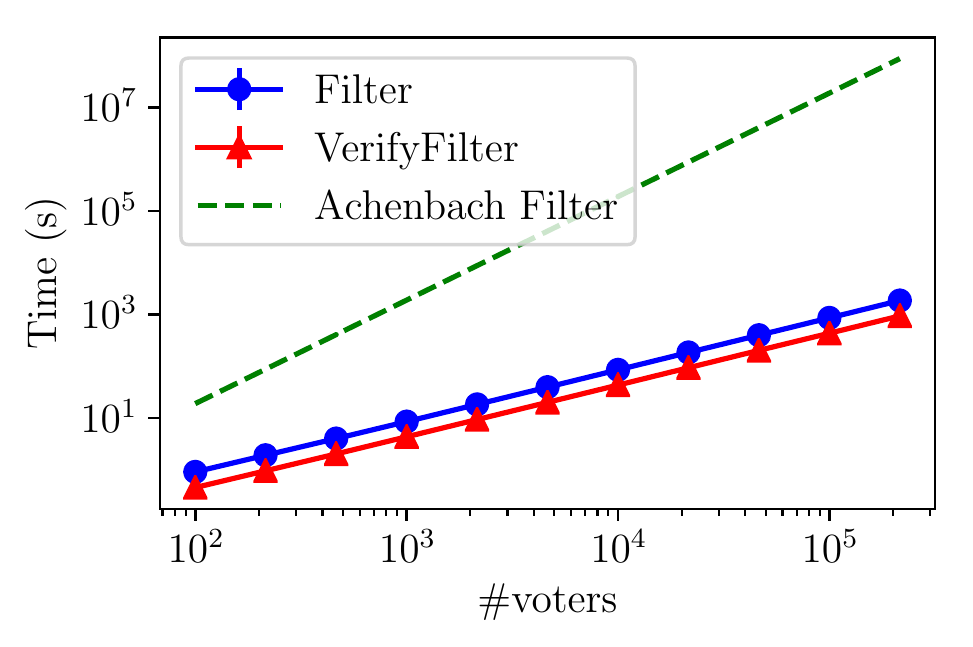}
  \includegraphics[width=0.33\textwidth]{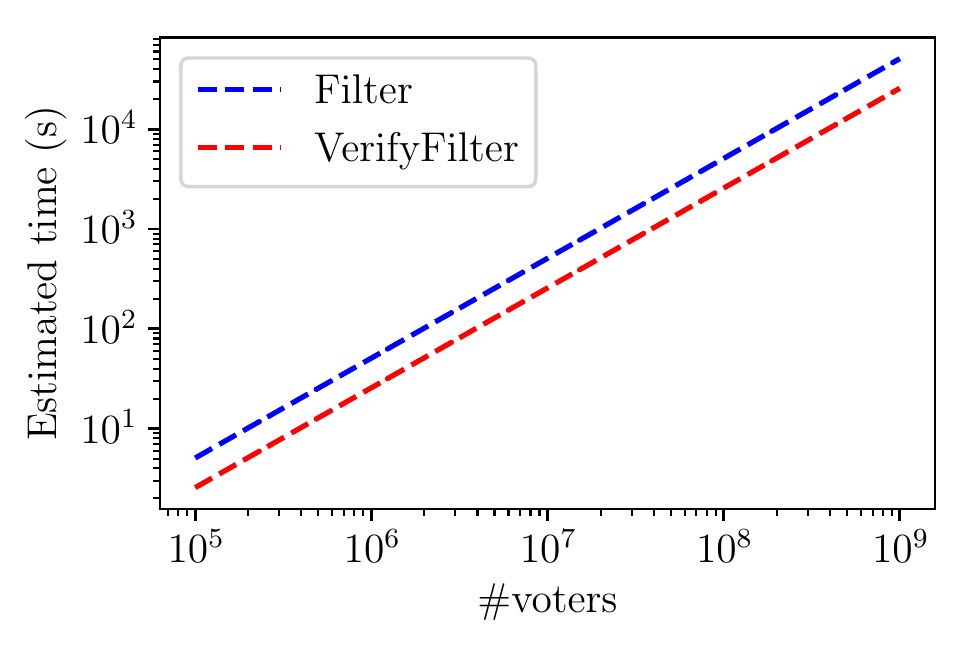}
  \includegraphics[width=0.33\textwidth]{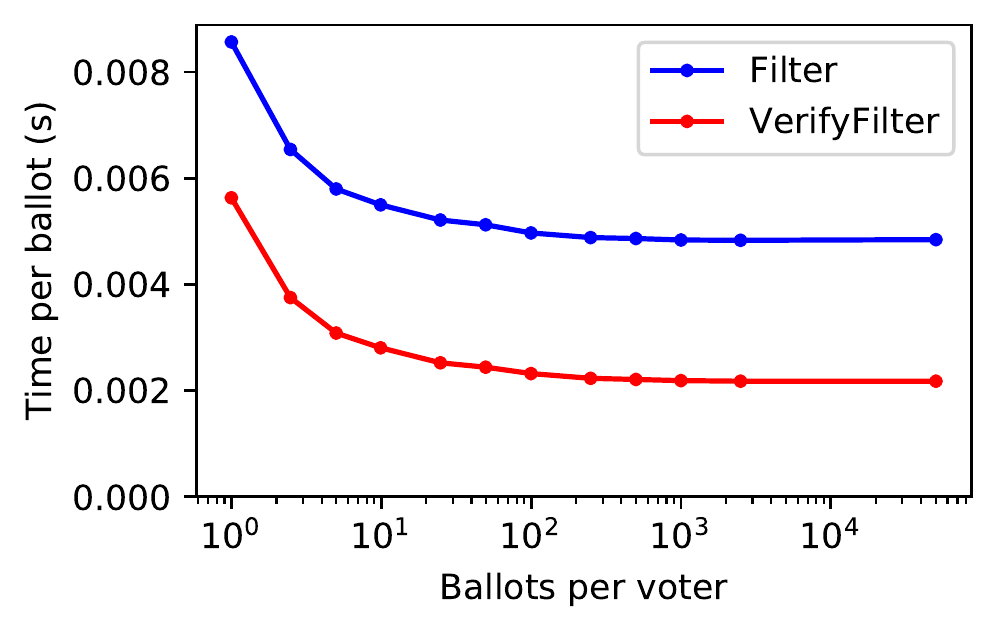}
  \vspace{-1cm}
  \caption{\label{fig:filter-tally-time} Cost of \filter, \tally,
  \verifyfilter and \verifytally: Measured cost on single core (left);
  estimated cost on 8 processor machine ($8\times28$ cores, center); and
  effect of different distributions of 50\,000 ballots (including
  dummies) among voters (right). Note that one ballot per voter causes
  the highest processing time.
}
\end{figure*}

We evaluate the performance of \name
using a Python prototype implementation of its core cryptographic operations. 
We did not implement the \gettoken protocol, but note that it can be implemented
easily and cheaply using standard cryptography. We also did not implement the
bulletin board as it is not core to our design.
We use the
\textsf{petlib}~\cite{Petlib} binding to
OpenSSL for the group operations using the fast NIST
P-256 curve. We ran all experiments in Linux on a single core of an Intel i3-8100
processor running at 3.60GHz. We
expect nation-wide elections to have much more processing power available. For
example, the Swiss CHVote system, which aims to support 8 million voters, has
around 32 cores available per party in the system. We also include
performance \emph{estimates} of running the system on a large machine with 8 Intel
Xeon Platinum 8280L processors, with 28 cores each, running at 2.7Ghz. 
As our scheme is almost completely parallelizable (only the hash functions for
the non-interactive zero-knowledge proofs need to be computed sequentially), we
estimate a 90\% parallelization gain: a speedup of 170 times when using the 8x28
cores with respect to the single core.

For all experiments we empirically select the best cover size $\coversizebase$ by sweeping over values from 1 to 64. In
the majority of cases the optimal $\coversizebase$ is in the range $[2,4]$.

\para{Creating a ballot.} We use an ElGamal ciphertext to
encrypt the voter's choice, and a Bayer and Groth \cite{10.1007/978-3-642-38348-9_38}
zero-knowledge proof of membership to show that the selected candidate is eligible.
Creating a ballot from 1000 eligible candidates costs 1.6 seconds, 
while verifying its correctness costs 0.24 seconds. The size of this proof is
1.5\,kB.
% 4 * log D points in the curve and 3 * log D elements of Zp, where D is the size of the set.

\para{Impact of revoting.} 
Figure~\ref{fig:dummies-overhead} shows the overhead, in terms of
number of dummies per real ballot depending on the number of votes. 
This overhead influences the computation time of 
shuffling and filtering in the tally phase. 
We consider different revoting behavior.
In the leftmost figure we model this behaviour as percentage of the number of voters: 50\% models that
half of the voters revoted once, and 200\% models that all voters
vote twice. We note that the overhead of 100\% voters revoting once is equivalent
to, for example, 25\% of the voters revoting 4 times.
As expected, the overhead increases with both
the number of voters and the number of revoted ballots.
However, even for 100 million voters revoting twice (200\% revotes), the
overhead is at most a factor of 32 (Figure~\ref{fig:dummies-overhead} left).

However, casting a vote takes time. Thus, revoting patterns are constrained
by the number of ballots that can be cast during an election. We consider
an election period of 24h (larger than most countries), and bound
how often a \emph{single} voter can vote (1 ballot per second, per ten seconds, and per minute).
As this limits the number of voters with a large amount of ballots,
we do not need large covers, reducing the overhead (see 
Figure~\ref{fig:dummies-overhead}, center). Similarly, assuming that 
all voters will revote is very conservative. In a normal election one expects 
the vast majority of voters to vote once. In Figure~\ref{fig:dummies-overhead}, right,
we show the overhead when the number of voters that
cast more than one vote is limited. As fewer voters revote, the total amount
of votes is smaller and so are the covers.

\para{Filtering.}  
We implemented a non-optimized version of Bayer-Groths verifiable shuffle 
protocol~\cite{Bayer2012} to implement steps~\ref{proc:shuffle} and ~\ref{proc:filter:second-shuffle} of 
Procedure~\ref{proc:filter}.
% Moreover, we implemented non-interactive proofs of correct decryption and correct 
% reencryption (steps~\ref{proc:filter:decrypt-vid-index}
% and~\ref{proc:filter:reencryption} of Procedure \ref{proc:filter} respectively).
We measure the execution time of filtering and
verifying, when varying the number of voters.
Figure~\ref{fig:filter-tally-time} left shows the times to run \filter and \verifyfilter on a single core machine.
Figure~\ref{fig:filter-tally-time} middle shows the estimated processing times on the big 8 processor Xeon machine.
We estimate that the 8 processor machine can filter and tally the second round presidential election in Brazil (147 million registered voters) in 95 minutes if no voter revotes,
and within a day assuming 50\% extra ballots and at most one ballot per voter per ten seconds. 
We note that elections usually tally
ballots per state, city, or smaller electoral district. Thus, in general we expect the number of ballots to be much smaller. 
All ballot groups in Figure~\ref{fig:filter-tally-time} left and center have size one.
Figure~\ref{fig:filter-tally-time} right shows the effect of larger ballot groups resulting
from revoting and dummy voters. As the average group size
increases, the computation time goes down. Therefore,
Figure~\ref{fig:filter-tally-time} gives an upper bound on the processing time,
given a known cover size.

For comparison we computed a lower bound on the filter cost of Achenbach et
al.'s filter method by counting the number of group operations needed per
ballot. We used this number to compute the estimate in
Figure~\ref{fig:filter-tally-time} left. A small-town election with
100.000 ballots takes 8 core months to filter in their scheme. Even on the large
Xeon machine, an election with 1 million ballots takes over four months to complete.
Our method needs respectively 10 core minutes and 30 seconds. The sizes of the tally
proofs in \name for these examples are 54 and 501 MB respectively. 
% For 100.000
% Shuffle with m = 8, 4.3 MB
% Proof decryption 513b (1 point and 1 Zp) per proof
% Proof re-encryption 3334 (6 points and 7 Zp) per proof. 
% Total 54 MB
% For 1M
% Shuffle with m = 8, 20 MB
% Other proofs = 481MB
% Total 501MB 

\parait{Smaller regions.} Many countries report election results per region,
such as a province, a city, or a neighborhood. In those cases, results can be
computed per region at lower computation cost. However, even in this setting,
Achenbach et al.'s quadratic approach scales poorly. We note that the allowable
size of reporting regions depend on local regulations, with the smallest regions
likely being cities or neighborhoods, which can easily total 100.000s of
voters. As Figure~\ref{fig:filter-tally-time} (left) shows, even in this
configuration, the quadratic approach requires 3 to 4 orders of magnitude more
computation resources than \name.

\para{Tallying.} We also measured the execution time of a single step of the mix
network -- a single shuffle and one verifiable decryption -- using our
verifiable shuffle implementation. Our results show that one step is a factor of three times
faster than our filter protocol, e.g., mix-and-decrypting the 100.000
ballots takes around 3 core minutes and 1 million ballots takes 10 seconds on the Xeon
machine.
% !TEX root = ../open_filtering.tex

\section{Conclusion}
\label{conclusion}

Due to its complexity and cost, coercion resistance has been often overlooked in
remote voting schemes. We introduced \name, a revoting scheme that enables
cleartext filtering thanks to efficient deterministic padding. \name
does not require users to store cryptographic material, and can efficiently
handle millions of votes. We provided a new coercion resistance definition and
updated existing definitions for ballot privacy and verifiability to the
revoting setting. We have proven that \name satisfies all of them.

\bibliographystyle{plain}
\bibliography{open_filtering}

\appendix

\section{Proof of ballot privacy, strong correctness and strong consistency}
\label{appendixA}
\label{app:proof-ballot-secrecy}

\begin{proof}[Proof of \cref{thm:ballot-secrecy}]
  This proof is very similar to the proof of ballot privacy of Helios in the
  full version of Bernhard et al.~\cite{BernhardCGPW15}. We start with the adversary playing the ballot privacy game with $\gamebit = 0$ and after a sequence of game steps transitions, the adversary finishes playing the ballot privacy game with $\gamebit = 1$. We argue that each of these steps are indistinguishable, and therefore the results follows.

  The proof proceeds along the following sequence of games:
  \begin{description}
  \item[Game $G_0$.] Let game $G_{0}$ be the $\bprivexp{0}{\adv}{\mathcal{V}}$ game (see Figure~\ref{fig:bpriv-game} and Definition~\ref{def:bpriv}).
  \item[Game $G_1$.] Game $G_1$ is as in $G_0$ but we now compute
    \begin{equation*}
     \tallyproofdec_0 = \simproof(\BB_0, \result)
    \end{equation*}
    by simulating the proof using the random oracle instead of using the real proof from $\tally(\BB_0, \votesk).$  Because of the simulation properties of the zero-knowledge proof system, \adv cannot distinguish these two games.
  \item[Game $G_2$.] As in game $G_1$, but now $\oracletally(\selectedvotelist,
    \filteradditions)$ ignores $\selectedvotelist$ and $\filteradditions$ provided by $\adv$ when computing the result $\result$. In particular, $\oracletally$ now proceeds as follows: \\[2mm]
    \begin{tabular}{@{}ll@{}}
      \toprule
        \multicolumn{2}{@{}l}{$\oracletally(\selectedvotelist, \filteradditions)$} \\
        & If $\verifyfilter(\BB_{\gamebit}, \selectedvotelist, \filteradditions) = \bot$ return $\bot$ \\
        & $(\result, \tallyproofdec_0) \gets \tally(\BB_0 \parallel \filter(\BB_0, \lastballot', \tssk), \votesk)$ \\
        & $\BB_{\gamebit} \gets \BB_{\gamebit} \parallel \selectedvotelist \parallel \filteradditions$ \\
        & $\BB_{1 - \gamebit} \gets \BB_{1 - \gamebit} \parallel \filter(\BB_{1 - \gamebit}, \lastballot',  \ tssk)$ \\
        & $\tallyproofdec_0 = \simproof(\BB_0, \result)$ \\
        & $\tallyproofdec_1 = \simproof(\BB_1, \result)$ \\
        & return $(\result, \tallyproofdec_{\gamebit})$ \\
      \bottomrule
    \end{tabular}\\[3mm]
    The proofs included in $\filteradditions$ ensure that \adv honestly computed the filtering step. Therefore, the adversary's view is indistinguishable from that in $G_1$.
  \item[Game $G_3$.] As in game $G_2$, but in $\oracleboard$ we return $\BB_1$. Note that in $G_3$ the adversary has the same view as in the $\bprivexp{1}{\adv}{\mathcal{V}}$ game. All that is left to show is that $G_2$ and $G_3$ are indistinguishable.
  \end{description}

  We now show that no adversary $\adv$ can distinguish $G_2$ from $G_3$. Let $\nroraclevotes$ be the number of $\oraclevote(\votetoken, \candidate_0, \candidate_1)$ calls that the adversary $\adv$ made. In particular, for the $i$th call to $\oraclevote$, remember the tuple $(\ballot_0, \ballot_1, \candidate_0, \candidate_1)$ of candidates and resulting ballots. We now build a series of games $H_0, \ldots, H_{\nroraclevotes}$ and proceed by a hybrid argument.

  In game $H_i$ we show to the adversary a bulletin board where the first $i$ ballots cast using $\oraclevote$ on $\BB_0$ are replaced by those of $\BB_1$.
  More precisely, in all games $H_i$ we keep track of an additional bulletin board $\BB$ that is shown to the adversary, i.e., $\oracleboard$ now returns $\BB$. Whenever the adversary makes an $\oraclecast(\ballot)$ query, we also add $\ballot$ to $\BB$, i.e., $\BB \gets \BB \parallel \ballot$. In game $H_i$ in response to the first $i$ calls to $\oraclevote$, we additionally set $\BB \gets \BB \parallel \ballot_1$. For the remaining calls we additionally set $\BB \gets \BB \parallel \ballot_0$. Note that $H_0 = G_2$ and that $H_{\nroraclevotes} = G_3$.

  From Bernhard et al.'s work on Helios~\cite{Bernhard2012} we know that in the random oracle model under the DDH assumption the ballot encryption scheme based on ElGamal with a non-interactive proof of correct construction is NM-CPA secure, that is,
  \begin{equation*}
    \left| \textrm{Pr}\left[ \nmcpaexp{0}(\secpar) = 1 \right] -
           \textrm{Pr}\left[ \nmcpaexp{1}(\secpar) = 1 \right]
    \right|
  \end{equation*}
  is negligible in $\secpar$, where $\nmcpaexp{\gamebit}$ is as in Figure~\ref{fig:nm-cpa}. We reduce to the NM-CPA security of the encryption scheme to show that $H_i$ is indistinguishable from $H_{i - 1}$.
  
\begin{figure}[tbp]
  \centering
  \begin{tabular}{@{}ll@{}}
    \toprule
    \multicolumn{2}{@{}l}{$\nmcpaexp{\gamebit}(\secpar)$:} \\
    & $(\pk, \sk) \gets \votekeygensimple(1^{\secpar})$ \\
    & $m_0, m_1 \gets \adv(\textrm{``find''}, \pk)$ \\
    & $c^* = \voteenc(\pk, m_{\gamebit})$ \\
    & $\gamebit' \gets \adv^{\oracledec}(\textrm{``guess''}, \pk, c^*)$ \\
    & Output $\gamebit'$ \\[2mm]
    \multicolumn{2}{@{}l}{$\oracledec(\vec{c})$:} \\
    & If $c^* \in \vec{c}$ then return $\bot$ \\
    & $\vec{m}_i = \votedec(\sk, \vec{c}_i)$ \\
    & Return $\vec{m}$ \\
    \bottomrule
  \end{tabular}
  \caption{In the NM-CPA experiment $\bprivexp{\gamebit}{\adv}{\mathcal{V}}$, the adversary \adv finds two messages $m_0$ and $m_1$ of which it asks an encryption of the challenger. It is then allowed to ask the decryption of a vector of ciphertexts $\vec{c}$ of its decryption oracle $\oracledec$. It may only call this oracle once.}
  \label{fig:nm-cpa}
\end{figure}
  To show this, we create an adversary $\mathcal{B}$ against NM-CPA. Internally, $\mathcal{B}$ uses adversary \adv. Adversary $\mathcal{B}$ receives the public key $\pk$ from its challenger. At the start of the game $\mathcal{B}$ runs $\setup$ as normal, but instead it sets $\votepk = \pk$. It then answers the $j$th $\oraclevote(\votetoken, \candidate_0, \candidate_1)$ query as follows:
  \begin{itemize}
  \item For $j < i$ it sets $\BB \gets \BB \parallel \ballot_1$
  \item For $i = i$ it returns $\candidate_0, \candidate_1$ to the NM-CPA challenger to receive a challenge ciphertext $c^*$, and uses that ciphertext when running $\vote$ to obtain a ballot $\ballot^*$ and set $\BB \gets \BB \parallel \ballot^*$.
  \item For $j > i$ it sets $\BB \gets \BB \parallel \ballot_0$
  \end{itemize}
  Thereafter $\mathcal{B}$ answers the $\oracletally$ query as follows. It cannot directly compute the tally, as it does not know the decryption key $\votesk$.
  However, it knows $\tssk$ so it can recompute the $\filter(\BB_0, \lastballot', \tssk)$ to determine which ballots $\ballot_{i_1},\ldots, \ballot_{i_\nrballotgroups}$ on $\BB_0$ should be included in the final tally (recall that the result is always computed on $\BB_0$, and that as per $G_2$ we do not use $\filteradditions$ provided by the adversary). Then proceed as follows:
  Let $\Gamma = (\encryptedvote_{i_1}, \voteproof_{i_1}), \ldots, (\encryptedvote_{i_{\nrballotgroups}}, \voteproof_{i_{\nrballotgroups}})$ be the corresponding vote ciphertexts and proofs. Then, $\mathcal{B}$ computes the result $\result$ as follows:
  \begin{itemize}
  \item If $c^* \in \Gamma$, then the ballot for candidate $\candidate = \candidate_0$ in query $i$ should be included in the tally as well. Recall that the tally always is computed over $\BB_0$, therefore, $\mathcal{B}$ sets
    $(\candidate_{i_1}, \ldots, \candidate_{i_{\nrballotgroups - 1}}) = \oracledec(\Gamma \setminus \{c^*\})$ and sets
    \begin{equation*}
      \result = \partialtally(\candidate_{i_1}, \ldots, \candidate_{i_{\nrballotgroups - 1}}, \candidate_0).
    \end{equation*}
  \item Otherwise, $\mathcal{B}$ sets $(\candidate_{i_1}, \ldots, \candidate_{i_{\nrballotgroups}}) = \oracledec(\Gamma \setminus \{c^*\})$ and sets
    \begin{equation*}
      \result = \partialtally(\candidate_{i_1}, \ldots, \candidate_{i_{\nrballotgroups}}).
    \end{equation*}
  \end{itemize}
  Finally, as in game $G_2$, $\mathcal{B}$ simulates the tally proof. Note that if $\gamebit = 0$ in $\mathcal{B}$'s NM-CPA game, then $\mathcal{B}$ perfectly simulates $H_{i - 1}$, and if $\gamebit = 1$ then it perfectly simulates $H_{i}$. Therefore, any distinguisher between $H_{i}$ and $H_{i - 1}$ breaks the NM-CPA security of the voting scheme.

  A standard hybrid argument now shows that $H_0 = G_2$ is indistinguishable from $H_{\nroraclevotes} = G_3$. This completes the proof.
\end{proof}

\subsection{Strong Consistency}
The ballot privacy definition ensures that ballots and the proof of correct
tally $\tallyproofdec$ do not leak anything about how voters voted. However,
maliciously crafted voting schemes might leak information about honest votes in
the result $\result$ itself. To ensure that this is not possible, Bernhard et
al.~\cite{Bernhard:2015:SCA:2867539.2867668} introduced the notion of strong
consistency. Intuitively, this notion ensures that the result $\result$ is equal to the result function applied directly to the valid ballots (skipping the filter and tally phase). We follow the exposition of Bernhard et al., but make some changes to account for the fact that our scheme selects ballots with the highest corresponding ballot number $\ballotnumber$, rather than simply the last per voter.

Our voting scheme depends on a formal result function
$\resultfunction:((\vidspace\times\mathbb{N})\times\candidatelist)^*\rightarrow
\resultspace$, where $\vidspace$ is the space of voters identifiers, and
$\resultspace$ is the result space. Our result function selects, for every $\vid \in \vidspace$, the ballot $((\vid, \ballotnumber), \candidate)$ where $\ballotnumber$ is the maximal counter for this voter. Then it counts the number of votes per candidate $\candidate$ in the selected ballots and returns the result.

To model that the result $\result$ output by \tally is consistent with the result function $\resultfunction$, we require the existence of an extraction algorithm $\extractoralgorithm$ that takes as input the TS's key $\tssk$, the trustee key $\votesk$ and a ballot, and outputs a tuple $((\vid, \ballotnumber), \candidate) \in ((\vidspace\times\mathbb{N})\times\candidatelist)$ with the corresponding voter identifier $\vid$, ballot number $\ballotnumber$ and candidate $\candidate$ in this ballot. If it fails to extract these values, it outputs $\bot$.

Moreover, we require a method \validindividual that validates ballots
independent of the bulletin board. The function $\validindividual$ takes as input the election public key $\pk$ and a ballot, and outputs $\top$ if the ballot is valid, and $\bot$ otherwise.

\begin{definition}[Adapted from Bernhard et al.~\cite{Bernhard:2015:SCA:2867539.2867668}]\label{def:strong-consistency}
  A voting scheme $\votingscheme = (\setup,\allowbreak \gettoken,\allowbreak \vote,\allowbreak \valid,\allowbreak \filter,\allowbreak \verifyfilter,\allowbreak \tally, \verify)$ for an electoral roll $\electoralroll$ and candidate list $\candidatelist$ has \textit{strong consistency} with respect to a result function $\resultfunction:((\vidspace\times\mathbb{N})\times\candidatelist)$
  if there exists algorithms $\extractoralgorithm$ and $\validindividual$ as above, such that the following three conditions hold:
  \begin{enumerate}
  \item \label{def:extractor}
  For any $(\pk, \pask, \tssk, \votesk)$ output by $\setup$, for all voters $i \in \electoralroll$ with voter identifier $\vid$, for all $\votetoken\gets\gettoken(i)$ where $\ballotnumber$ is the corresponding ballot number, and for any ballot $\ballot\gets\vote(\votetoken, \candidate)$ with $\candidate\in\candidatelist$,
  we have that $\extractoralgorithm(\votesk, \pask, \ballot) = ((\vid, \ballotnumber), \candidate).$
  \item \label{def:individual-validator}
  For any $(\BB, \ballot) \gets \adv()$ we have that $\valid(\BB, \ballot) = \top$ implies $\validindividual(\ballot) = \top$.
  \item \label{def:scons-game-def}
  For all probabilistic polynomial time adversary $\adv$ we have that
  \begin{equation*}
  \textrm{Pr}\left[\sconsexp(\secpar, \electoralroll, \candidatelist) = 1 \right]
  \end{equation*}
  is a negligible function in $\secpar$ (see Figure~\ref{fig:scons} for the game).
  \end{enumerate}
\end{definition}

The first condition ensures that $\extractoralgorithm$ can extract $((\vid, \ballotnumber), \candidate)$ correctly for honestly created ballots. The second condition ensures that ballots that are accepted by $\valid$ with respect to the board $\BB$ must also be accepted by $\validindividual$. Finally, the third condition ensures that the adversary cannot produce bulletin boards where the result $\result$ does not correspond to the formal result function $\resultfunction$ executed on the individual ballots. (The adversary loses if \filter or \tally aborts because of an invalid bulletin board.)

\begin{figure}[tbp]
  \centering
  \begin{tabular}{@{}ll@{}}
    \toprule
    \multicolumn{2}{@{}l}{$\sconsexp(\secpar, \electoralroll, \candidatelist)$:} \\
    & $(\pk, \pask, \tssk, \votesk) \gets \setup(1^{\secpar}, \electoralroll, \candidatelist)$ \\
    & $\BB = [\ballot_1, \ldots, \ballot_{\lastballot}] \gets \adv(\pk, \pask)$ \\
    & If $\exists \ballot_i$ s.t. $\validindividual(\ballot_i) = \bot$ then
      return 0 \\
    & Let $\selectedvotelist, \filteradditions \gets \filter(\BB, \lastballot', \tssk)$ \\
    & Let $(\result, \tallyproofdec) \gets \tally(\BB \parallel \selectedvotelist \parallel \filteradditions, \votesk)$ \\
    & If $\result = \bot$ return 0 \\
    & If $\result\neq\resultfunction(\extractoralgorithm(\votesk, \pask, \ballot_1), \ldots,\extractoralgorithm(\votesk, \pask, \ballot_{\lastballot})) $ \\
    & return 1, else return 0. \\
    \bottomrule
  \end{tabular}
  \caption{In the strong-consistency experiment
    $\sconsexp$, adversary \adv must output a board $\BB$ with ballots that are not tallied correctly given $\extractoralgorithm$.}
  \label{fig:scons}
\end{figure}

\subsection{Strong correctness}
Finally,
a malicious protocol designer might modify which ballots are accepted based on
earlier ballots.
To address this attack, Bernhard et al.~\cite{Bernhard:2015:SCA:2867539.2867668} introduce the notion of strong correctness. Informally, a scheme has strong correctness if honestly generated ballots are accepted regardless of the content of the bulletin board.

\begin{figure}[tbp]
  \centering
  \begin{tabular}{@{}ll@{}}
    \toprule
    \multicolumn{2}{@{}l}{$\scorrexp(\secpar)$:} \\
    & $(\pk, \pask, \tssk, \votesk) \gets \setup(1^{\secpar})$ \\
    & $(i, \votetoken, \candidate, \BB) \gets \adv(\pk, \pask)$ \\
    & Let $\ballot = \vote(\votetoken, \candidate)$ \\
    & If $\vote$ aborts because $\votetoken$ is invalid, return $\top$ \\
    & Else return $\valid(\BB, \ballot)$  \\
    \bottomrule
  \end{tabular}
  \caption{In the strong-correctness experiment
    $\scorrexp$, the adversary outputs a board $\BB$ such that an honest ballot by a voter of its choice is not valid.}
  \label{fig:scorr}
\end{figure}
\begin{definition}[Adapted from Bernhard et al.~\cite{Bernhard:2015:SCA:2867539.2867668}]\label{def:strong-correctness}
Consider a voting scheme $\votingscheme = (\setup, \gettoken, \vote, \filter, \verifyfilter, \tally, \verify)$ for an electoral roll $\electoralroll$ and candidate list $\candidatelist$. We say the scheme has \emph{strong correctness} if
\begin{equation*}
\textrm{Pr}\left[\scorrexp(\secpar) = \bot\right]
\end{equation*}
is a negligible function in $\secpar$ (see Figure~\ref{fig:scorr} for the game).
\end{definition}

\begin{theorem}\label{thm:strong-consistency-correctness}
\name provides strong-consistency and strong-correctness.
\end{theorem}

\begin{proof}[Proof of \cref{thm:strong-consistency-correctness}]
This proof roughly follows that of the strong consistency and strong correctness
of Helios in the full version of Bernhard et al.~\cite{BernhardCGPW15}. To show that \name is strongly consistent, we define the following $\extractoralgorithm$ and $\validindividual$ algorithms:
\begin{itemize}
\item $\extractoralgorithm(\ballot, \tssk, \votesk)$ operates on a ballot
  $\ballot = (\encryptedvote, \voteproof, \pk, \encryptedvid, \encryptedindex, \tokensignature, \signature)$. First, it verifies the proof $\voteproof$, and the signatures $\tokensignature$ and $\signature$ as in step~\ref{proc:valid:checks} of \valid in \cref{proc:valid}. If any check fails, it returns $\bot$. Otherwise, it recovers the candidate $\candidate = \votedec(\votesk, \encryptedvote)$ (note that $\candidate \in \candidatelist$ because $\voteproof$ is valid). Then it decrypts $\encryptedvid$ and $\encryptedindex$ to get $(\vid, \ballotnumber)=(\ecdec(\pask, \encryptedvid), \ecdec(\pask, \encryptedindex))$. It returns $((\vid,\ballotnumber), \candidate)$.
\item $\validindividual(\ballot)$ proceeds as in step~\ref{proc:valid:checks} of \valid in \cref{proc:valid} to verify the ballot:
\begin{align*}
      & \voteverify(\votepk, \encryptedvote, \voteproof) \\
      & \signverify(\pk, \signature, \encryptedvote \parallel \voteproof \parallel \pk \parallel \encryptedvid \parallel \encryptedindex \parallel \tokensignature)\\
      & \signverify(\papk, \tokensignature, \pk \parallel \encryptedvid \parallel \encryptedindex).
      \end{align*}
It returns $\top$ if all are valid, and $\bot$ otherwise.
\end{itemize}

First we show that the first condition of strong consistency is satisfied. By the correctness of the zero-knowledge proofs and decryption algorithms, \extractoralgorithm will indeed extract the required values for valid ballots.

Since \validindividual executes a strict subset of the checks in \valid, it follows that the second condition is trivially satisfied.

For the third condition, we need to show that the adversary cannot create a
valid bulletin board \BB (i.e., one on which \filter and \tally do not fail),
but where the result is incorrect (respect to the output calculated with the extractor
function).

Note that by the checks in steps \ref{proc:filter:checks} and
\ref{proc:filter:decrypt-vid-index} of \cref{proc:filter}, we know that the identifier pairs 
$(\decryptedvid{}, \decryptedindex{})$ are unique. Consider the group
$\ballotgroup_j$ of ballots corresponding to $\decryptedvid{j}$. The ideal result
function $\resultfunction$ includes the vote where the ballot index is highest.
In exactly the same way, \filter sets $\preselectedvote{j}$ to
$\encryptedvote_{j*}$ where the index $j*$ maximizes the ballot index
$\decryptedindex{j*}$. The equivalence of the ideal result and the result
produced by tally now follows.

To show that \name is strongly correct we need to prove that an adversary cannot
create a ballot box $\BB$ such that an honest voter, when generating an honest ballot
$\ballot$, that ballot will be rejected, i.e., $\valid(\BB, \ballot) =
\bot$. Note that the verification in $\valid(\ballot)$ is twofold.
First, it verifies the validity of the ballot. It is trivial to see that
this check passes for an honestly generated ballot. Second, it checks that the
ephemeral public key $\pk$ and encrypted vote $\encryptedvote$ do not yet appear
on the bulletin board. Clearly, $\encryptedvote$ does not appear because it was
just generated honestly by the user. Moreover, neither does the public key $\pk$
appear before, because it was just freshly generated by the PA. Given that these 
two values contain a source of randomness when generated, it proceeds that \adv
can only win with negligible probability.
\end{proof}

\section{Proof of Coercion Resistance}
\label{app:proof-coercion-resistance}

\begin{proof}[Proof of \cref{thm:coercion-resistance}]
  We first specify how to construct $\simproof$ and $\simfilter$. As in the
  ballot privacy proof, $\simproof(\BB, \result)$ simply simulates the proof of
  shuffle and the proof of
  correct decryption in $\tally$, so that regardless of the values in
  $\selectedvotelist$, $\result$ is the correct outcome.

  The algorithm $\simfilter(\BB, \lastballot', \result)$ proceeds similarly. It takes as input
  the bulletin board $\BB$, which it uses to determine the number of ballots $\lastballot$,
  the number of registrations $\lastballot'$, and the result $\result$.
  Moreover, it derives the number of real voters $\lastvoter$ using $\result$.
  It uses these data to compute the cover, and it adds the correct number of
  dummy ballots (for these, it sets $\encryptedvid$ and $\encryptedindex$ to
  random ciphertexts) to obtain $\strippedballotlist$. Then it computes a list
  of zero ciphertexts (encryptions of zero) of equal length, and simulates the
  shuffle proof $\shuffleproof$.
  It then generates fake voter identifiers $\decryptedvid{}$ and
  $\decryptedindex{}$ corresponding to the cover it computed earlier, associates
  these to shuffled ballot $\ballot_i$, and simulates the proofs
  $\decryptproof_i$. Next, for each resulting
  group, it generates a random encryption of zero $\preselectedvote{j} = \votezeroenc(\votepk, r_j)$ and 
  simulates the corresponding proof $\reencryptionproof_j$. Then, it returns the 
  randomness $r_j$ and the indices of the dummy voters corresponding to the 
  cover it computed early. Finally, for each remaining vote, 
  it generates a random $\selectedvote{j}$ and 
  simulates the shuffle proof $\shuffleproof'$. 
  
  In this proof, we will step by step replace all the ciphertexts that depend on
  the bit $\gamebit$ by random ciphertexts. In particular, we first show that the
  adversary learns nothing about $\gamebit$ during the election phase. We then
  show that it also learns nothing about $\gamebit$ during the tally phase. The
  result follows.

  \begin{enumerate}[label = {Game $G_{\arabic*}$.}, ref = {$G_{\arabic*}$}, labelwidth=\widthof{Game $G_{9}$},leftmargin=\parindent,labelindent=0pt,itemindent=!]
  \item \label{cr:game:base}
    Game \ref{cr:game:base} is as the $\crexp{\gamebit}{\adv}{\mathcal{V}}(\secpar, \electoralroll, \candidatelist)$ experiment. (Note that contrary to the proof of ballot privacy we do not fix the value for $\gamebit$.)
  \item \label{cr:game:direct-tally}
    Game \ref{cr:game:direct-tally} is as game \ref{cr:game:base}, but we compute the result directly based on the ballots on $\BB_0$. Let $[\ballot_1, \ldots, \ballot_{\lastballot}]$ be the list of ballots where
    \begin{equation*}
    \ballot_i = (\encryptedvote_i, \voteproof_i, \pk_i, \encryptedvid_i, \encryptedindex_i, \tokensignature_i, \signature_i ).
    \end{equation*}
    Let $\candidate_i = \votedec(\sk, \encryptedvote_i)$, $\vid_i = \ecdec(\tssk, \encryptedvid_i)$, and $\ballotnumber_i = \ecdec(\tssk, \encryptedindex_i)$. Then compute the result:
    \begin{equation*}
      \result = \rho( ((\vid_1, \ballotnumber_1), \candidate_1), \ldots, ((\vid_{\lastballot}, \ballotnumber_{\lastballot}), \candidate_{\lastballot}))
    \end{equation*}
    As per strong consistency, games~\ref{cr:game:direct-tally} and~\ref{cr:game:base} are indistinguishable.
  \item \label{cr:game:simulated-proofs}
    Game \ref{cr:game:simulated-proofs} is as game \ref{cr:game:direct-tally},
    but with all the zero-knowledge proofs replaced by simulations.
    This includes the shuffle proof $\shuffleproof$, the decryption proofs of $\decryptproof_i$ of the shuffled $\encryptedvid_i'$ and $\encryptedindex_i'$s, the reencryption proofs $\reencryptionproof_i$, and the shuffle proof $\shuffleproof'$ produced in \filter; as well as the tally proof $\tallyproofdec_0$ which we replace by the output of $\simproof(\BB_0, \result)$. We use the random oracle to simulate this step, which is indistinguishable by the simulatability of the zero-knowledge proof system.
  \item \label{cr:game:no-label-decrypt}
    Game \ref{cr:game:no-label-decrypt} is as game
    \ref{cr:game:simulated-proofs} but we do not decrypt the $\encryptedvid_i$
    and $\encryptedindex_i$ anymore when running \filter. Instead, we proceed as
    follows. All ballots $\ballot_i = (\encryptedvote_i, \voteproof_i, \pk_i,
    \encryptedvid_i, \encryptedindex_i, \tokensignature_i, \signature_i )$ on
    the bulletin boards are valid. Hence, $\signature_i$ is a valid signature by
    $\PA_{0}$ resp. $\PA_{1}$ on $\encryptedvid_i$ and $\encryptedindex_i$.
    Since the signature scheme is unforgeable, we know these ciphertexts were
    created by $\PA_{0}$ resp. $\PA_1$. Hence, we can associate to them the
    corresponding plaintexts $\vid_i$ and $\ballotnumber_i$. Moreover, we know
    the permutation used by the TS during \filter, so we can also provide the
    correct plaintexts in step \ref{proc:filter:decrypt-vid-index} of \filter
    on $\BB_0$ (recall the proofs of decryption $\decryptproof_i$ are already 
    simulated).
  \item \label{cr:game:labels-random}
    Game \ref{cr:game:labels-random} is as game~\ref{cr:game:no-label-decrypt},
    but we replace the ciphertexts $\encryptedvid_i$ and $\encryptedindex_i$ in
    the token $\votetoken_i$ by random ciphertexts for all tokens. Similarly, we replace the $\encryptedvid_i$ and $\encryptedindex_i$ ciphertexts for the dummy ballots by random ciphertexts. Note that per the change in game~\ref{cr:game:no-label-decrypt} we still associate the correct plaintexts $\vid_i$ and $\ballotnumber_i$ in the \filter protocol. A hybrid argument with reductions to the CPA security of the ElGamal encryption scheme shows that games \ref{cr:game:labels-random} and \ref{cr:game:no-label-decrypt} are indistinguishable.
    This reduction is possible since we no longer need to decrypt these ciphertexts.
  \item \label{cr:game:zero-votes}
    Game \ref{cr:game:zero-votes} is as game~\ref{cr:game:labels-random}, but we replace the encrypted votes $\encryptedvote_i$ in the $\oraclevote()$ call by encryptions of the zero vector, i.e., $\encryptedvote_i = \votezeroenc(\votepk, r)$ for a uniformly random randomizer $r$. As in the ballot privacy proof, a hybrid argument with a reduction to the NM-CPA security of the ElGamal encryption scheme with zero-knowledge proof shows that games \ref{cr:game:zero-votes} and \ref{cr:game:labels-random} are indistinguishable. Note that in this reduction we use the $\oracledec$ of the NM-CPA challenger to decrypt votes in the adversary-determined ballots before computing the result $\result$.
  \end{enumerate}

  Note that as of game~\ref{cr:game:zero-votes}, the adversary's view of the bulletin board before calling $\oracletally()$ is independent of the value of $\gamebit$. (The ballots resulting from the $\oraclevote$ call also contain a random ephemeral public key $\pk$ and the signatures $\tokensignature$ and $\signature$, but these are also independent of the actual voter selected.)

  We now proceed to show that the adversary also cannot learn anything from the
  output of \filter. Notice that, regardless of the value of $\gamebit$, the
  filter step is computed with the same number of voters $\nrcastvoters$, the same number of
  ballots $\lastballot$ and the same number of obtained tokens $\lastballot'$. Therefore, the output of $\filter$ applied to $\BB_0$ and that of
  $\simfilter$ applied to $\BB_1$ should be indistinguishable.
  In the following game steps we replace the ciphertexts after shuffling by
  zero-ciphertexts and show that these steps are indistinguishable for the
  adversary.

  \begin{enumerate}[resume, label = {Game $G_{\arabic*}$.}, ref = {$G_{\arabic*}$}, labelwidth=\widthof{Game $G_{9}$},leftmargin=\parindent,labelindent=0pt,itemindent=!]
  \item \label{cr:game:zero-after-shuffle}
    Game \ref{cr:game:zero-after-shuffle} is the same as game~\ref{cr:game:zero-votes}, but we replace the ciphertexts
    $\encryptedvid_i', \encryptedindex_i'$ and $\tstag_i'$ after shuffling by random encryptions of zero. We proceed as if they still decrypt to the correct values. Note that we already simulate the shuffle proof and decryption proofs. Again, a hybrid argument with reductions to the CPA security of the ElGamal encryption scheme shows that the games \ref{cr:game:zero-after-shuffle} and \ref{cr:game:zero-votes} are indistinguishable.
    This reduction is possible since we no longer need to decrypt these ciphertexts.
  \item \label{cr:game:zero-votes-after-shuffle}
    Game \ref{cr:game:zero-votes-after-shuffle} is the same as game~\ref{cr:game:zero-after-shuffle}, but we replace the shuffled encrypted votes $\encryptedvote_i'$ by random encryptions of zero. Similarly, we replace the randomizations, $\openedrandomizers$, of the votes corresponding to dummy voters by the corresponding new randomization. This causes the pre-selected votes $\preselectedvote{j}$ per group to be incorrect, but this does not matter as we simulate the second shuffle proof, $\shuffleproof'$, anyway. As before, the indistinguishability of this step follows from the NM-CPA security of the vote encryption scheme.
  \item \label{cr:game:zero-votes-after-second-shuffle}
  Game \ref{cr:game:zero-votes-after-second-shuffle} is the same as game~\ref{cr:game:zero-votes-after-shuffle}, but we replace the second shuffled votes $\selectedvote{j}$ by random encryption of zero. This causes the selected votes $\selectedvote{j}$ after the shuffle to be incorrect with respect to the result, but this does not matter as we simulate the proof of the tally. As before, the indistinguishability of this step follows from the NM-CPA security of the vote encryption scheme.
  \item \label{cr:game:sim-all}
    Game \ref{cr:game:sim-all} is as game~\ref{cr:game:zero-votes-after-second-shuffle}, but we replace the filter and tally proofs on $\BB_0$ by simulations: we set
    $(\selectedvotelist_0, \filteradditions_0) \gets \simfilter(\BB_0, \lastballot', \result)$
    and $\tallyproofdec_0 \gets \simproof(\BB_0, \result)$. Note that this
    difference is purely syntactic, as per the changes we made before, we
    already computed exactly the output of \simfilter on $\BB_0$ and the result
    $\result$.
  \end{enumerate}

  Clearly the resulting view is independent of $\gamebit$. And coercion
  resistance follows.
\end{proof}

\section{Proof of Verifiability}
\label{app:proof-verifiability}
\begin{proof}[Proof of \cref{thm:verifiability}]
  At the end of the filter procedure, the \TS (or in our case, the adversary)
  outputs a list of selected votes $\selectedvotelist = \selectedvote{1},
  \ldots, \selectedvote{\lastvoter}$, a proof $\filteradditions$, the
  result $\result$ and the tally proof $\tallyproofdec$.
  Because $\tallyproofdec$ verifies, we know that the result $\result$ is the addition of the votes contained in $\selectedvote{1}, \ldots, \selectedvote{\lastvoter}$.

  We first show that each encrypted vote $\selectedvote{j}$ contains a vote for a single candidate or an empty vote. Given that $\shuffleproof$ and $\shuffleproof'$ validate, we know that the vote $\selectedvote{j}$ corresponds to a group of ballot indices $\ballotgroup_j = {i_1, \ldots, i_{\nrballotsingroup_j}},$ corresponding to voter identifier $\vid_j$. Moreover, the index $j*$ is such that the decrypted index $\decryptedindex{j*}$ is maximal. Because $\reencryptionproof_j$ is valid, we know that either
  \begin{enumerate}
  \item $\selectedvote{j}$ is the reencryption of vote $\encryptedvote_{j*}$ and $\ecdec( \tstag_{j*}) = g^0$,
  \item $\selectedvote{j}$ is the encryption of zero and $\ecdec( \prod_{k = 1}^{\nrballotsingroup_j} \tstag_{i_k}) = g^{\nrballotsingroup_j}$.
  \end{enumerate}
  We now show that in the first case $\encryptedvote_{j*}$ must be the encryption of a single candidate. Because $\tstag_{j*} = g^0$ and the correctness of the shuffle proofs $\shuffleproof, \shuffleproof'$, we know that $\encryptedvote_{j*}$ originates from a valid ballot cast by a voter. This ballot included a proof that $\encryptedvote_{j*}$ is the encryption of a single candidate.

  Let $\nrcastvoters$ be the number of voters that requested a voting token and for which at least one ballot is included on the bulletin board. We argue that the number of non-zero ballots that is included in the tally equals $\nrcastvoters$. Let $\vid_{i_1}, \ldots, \vid_{i_\nrcastvoters}$ be the corresponding voter identifiers. Because of the correctness of the shuffle, $\shuffleproof$, there exists corresponding groups $G_{i_1}, \ldots, G_{i_\nrcastvoters}$ to these voter identifiers after shuffling.

  We show that any other group $G_\xi$ contributes an empty vote to the tally. Let $\vid_\xi$ be the corresponding voter identity. The adversary cannot forge signatures by the PA, so any ballot with voter identity $\vid_\xi$ was added as a dummy ballot with $\tstag_D$ as a tag. Therefore, in group $G_\xi$, each tag encrypts $g^1$, so the only possible disjunct in $\reencryptionproof_\xi$ is therefore the second, and thus $\preselectedvote{\xi}$ is the encryption of zero. 
  
  We show that only such encrypted votes may be removed after the second shuffle. 
  The adversary needs to find a random $r$ such that $\selectedvote{i} = \votezeroenc(\votepk, r)$. Given that the DL-assumption holds, the adversary can only find
  such $r$ if the underlying plaintext is zero with very high probability. Given that the 
  encryption of candidate zero is not a permitted option for real voters, and given the 
  correctness of $\voteproof$, only votes corresponding to the above groups 
  may be removed by the adversary after the second shuffle. 

  So, we now know that only the groups $G_{i_1}, \ldots, G_{i_\nrcastvoters}$ each contribute exactly one candidate to the tally, and no more candidates are added by the other groups. We assign each group to one of the three groups in the game: the voters in $\textsf{Checked}$, the voters in $\textsf{Unchecked}$, or the voters in $\textsf{Corrupted}$. The result then follows.

  We now show that the correct values are tallied for each of the voters in $\textsf{Checked}$ that verified that their ballots were correctly cast. Consider a voter $i \in \textsf{Checked}$ with voter identifier $\vid_i$. Let $\ctr$ be the last ballot that it verified. We need to show that the tally includes either $i$'s ballot $\ctr$, or a later ballot. We know ballot $\ctr$ was added to the bulletin board. Therefore, the corresponding group $G_i$ (matching voter identifier $\vid_i$) containing $\nrballotsingroup_i$ ballots, must contain a shuffled ballot $(\encryptedvote_j', \vid_i, \decryptedindex{j}, \tstag_j')$ corresponding to the original ballot $\ctr$ (because the shuffle proof and decryption proofs are valid). Note that $\tstag_j'$ must be a decryption of $g^0$ by construction, therefore the tags in group $G_j$ (which must be encryptions of $g^0$ or $g^1$) can never sum to $g^{\nrballotsingroup_i}$ and therefore, we must take the first disjunct in the reencryption proof $\reencryptionproof_i$: $\preselectedvote{i}$ must be the reencryption of an encrypted ballot $j*$ where $\tstag_{j*}$ decrypts to $g^0$. Therefore, $\reencryptionproof_i$ must contain the encrypted vote corresponding to a real ballot cast by voter $i$. Finally, since $j*$ maximizes $\decryptedindex{j*}$ in the group, we know in particular, that $\decryptedindex{j*} \geq \decryptedindex{j}$ corresponding to the verified ballot. Therefore, we conclude that indeed $\preselectedvote{i}$ reencrypts either ballot $\ctr$ by voter $i$, or a later ballot by voter $i$. Given the correctness of the second shuffled proved with $\shuffleproof'$, there will be a selected ballot $\selectedvote{i}$ encrypting the same value as $\preselectedvote{i}$. Finally, given the correctness of the mixnet and decryption proofs in $\tally$, either $\ctr$ by voter $i$, or a later ballot by voter $i$, will be counted in the final tally. 
  
  Now suppose a group $G_i$ corresponds to a voter $i$ in $\textsf{Unchecked}$. Then, by the same argument as for voters in $\textsf{Checked}$, we know that the tally must either drop all ballots or include one of the ballots cast by voter $i$.

  Finally, any remaining groups correspond to voters in $\textsf{Corrupted}$. Notice that any voter that is not in $\textsf{Checked}$ or $\textsf{Unchecked}$ must be in $\textsf{Corrupted}$. Since, each remaining group corresponds to an actual voter, and this voter is not in either of the former groups, it must indeed correspond to a voter in $\textsf{Corrupted}$.
\end{proof}

\end{document}